\documentclass[ms,
copyedit]{informs4TD}

\usepackage{lmodern} %

\usepackage{graphicx, url}
\usepackage{graphics}
\usepackage{color}
\usepackage{tablefootnote}
\usepackage{multirow}
\usepackage{array}
\usepackage{mathrsfs}  
\usepackage{amssymb}
\usepackage{amsmath}
\usepackage{natbib}%
\usepackage{enumitem}

\usepackage{booktabs}
\usepackage{threeparttable}

\usepackage{subcaption}
\usepackage{tcolorbox}

\TheoremsNumberedThrough

\usepackage{footmisc}

\theoremstyle{TH} %

\newcounter{mainpropnum}
\theoremstyle{TH} %

\usepackage{longtable}

\newlength{\defbaselineskip}
\newcommand{\setlinespacing}[1]%
           {\setlength{\baselineskip}{#1 \defbaselineskip}}

\usepackage{natbib}
 \bibpunct[, ]{(}{)}{,}{a}{}{,}%
 \def\bibfont{\small}%

\setcitestyle{citesep={;}}

\makeatletter
\renewcommand*\env@matrix[1][c]{\hskip -\arraycolsep
  \let\@ifnextchar\new@ifnextchar
  \array{*\c@MaxMatrixCols #1}}
\makeatother

\usepackage{color}
\definecolor{hopkins-blue}{RGB}{0,45,114}
\definecolor{columbia-blue}{RGB}{185, 217, 235}
\definecolor{chicago-maroon}{RGB}{128,0,0}
\definecolor{cornell-red}{RGB}{179,27,27}
\definecolor{lawngreen}{RGB}{0,250,154}
\definecolor{uci-gold}{RGB}{254,204,7}
\definecolor{gray}{RGB}{192,192,192}

\usepackage{accents}

\usepackage{nicefrac}
\usepackage[colorlinks,citecolor=chicago-maroon,urlcolor=chicago-maroon,linkcolor=chicago-maroon]{hyperref}
\usepackage[nameinlink]{cleveref}
\usepackage{xurl}
\crefname{assumption}{Assumption}{Assumptions}
\crefname{lemma}{Lemma}{Lemmas}
\crefname{theorem}{Theorem}{Theorems}
\crefname{corollary}{Corollary}{Corollaries}
\crefname{proposition}{Proposition}{Propositions}
\crefname{claim}{Claim}{Claims}
\crefname{subclaim}{Subclaim}{Subclaims}
\crefname{procedure}{Procedure}{Procedures}
\crefname{algorithm}{Algorithm}{Algorithms}
\crefname{example}{Example}{Examples}
\crefname{figure}{Figure}{Figures}
\crefname{section}{Section}{Sections}
\crefname{appendix}{Appendix}{Appendices}
\crefname{table}{Table}{Tables}
\crefname{subproposition}{Proposition}{Propositions}

\newcommand{\PaperTitle}{Power Couple? AI Growth and Renewable Energy Investment}

\TITLE{\PaperTitle}

\RUNAUTHOR{Gui and Dai}
\RUNTITLE{\PaperTitle} %

\ARTICLEAUTHORS{
\AUTHOR{Luyi Gui{$^\ast$} \hspace{0.75in} 
Tinglong Dai{$^\dagger$}}
\medskip
\AFF{
{$^\ast$}Paul Merage School of Business, University of California, Irvine; Irvine, California 92697, email: \href{mailto:luyig@uci.edu}{\tt luyig@uci.edu}\\ \smallskip
{$^\dagger$}Carey Business School, Johns Hopkins University, Baltimore, Maryland 21202;
Data Science and AI Institute, Johns Hopkins University, Baltimore, Maryland 21218,
email: \href{mailto:dai@jhu.edu}{\tt dai@jhu.edu}}
}

\ABSTRACT{\OneAndAHalfSpacedXI \small
Artificial intelligence (AI) and renewable energy are increasingly being described as a ``power couple,'' based on the idea that rapid growth in AI will spur clean-energy investment. Yet growing AI demand could also deepen reliance on fossil power. We study when each outcome arises in a game in which renewable-capacity investment and AI scaling interact. The key is how the market value of greater AI capability grows relative to the energy needed to achieve it. When value grows at least as fast as energy use (market-led scaling), the developer pushes toward frontier capability even when additional electricity comes from fossil sources. Renewable investment can then enable further AI growth without eliminating fossil use. As climate damages increase, AI becomes more valuable for adaptation, strengthening incentives to sustain frontier capability despite the associated emissions. We call this the ``adaptation trap.'' When energy requirements grow faster than capability value (resource-led scaling), energy costs place greater limits on AI expansion. Renewable investment then makes additional capability less costly while also reducing emissions. As climate damages rise, the growing value of AI for adaptation can justify enough clean-capacity expansion to support AI entirely with renewable power. We call this the ``adaptation pathway.'' A calibrated case study shows that both mechanisms can arise at empirically plausible magnitudes. The results suggest that decarbonizing AI requires renewable capacity to keep pace with the growth of compute demand.  \looseness=-1
}


\OneAndAHalfSpacedXI

\begin{document}

\maketitle

\section{Introduction}

Over the past few years, artificial intelligence (AI) has emerged as both a significant climate challenge \emph{and} a promising climate solution.  Frontier training and large-scale inference consume substantial electricity: data centers consumed about 415 TWh of electricity in 2024, roughly 1.5\% of the global total, and their consumption is projected to more than double by 2030 \citep{iea2025}; training a single state-of-the-art model can require energy comparable to that used by thousands of households in a year \citep{iea2025}. At the same time, AI is increasingly deployed to support climate action: optimizing grids for renewable integration, improving weather and climate forecasts, and strengthening disaster response and other adaptation efforts \citep{Choi2024,rolnick2019tackling,stauffer2025aiCleanEnergyFuture}. This dual role has sharpened industry and policy debate over AI's net implications for climate change.

It is generally agreed that the climate impact of AI hinges on \emph{how} its energy demand is met: AI will exacerbate climate change if powered primarily by fossil fuels, but can be a solution if powered by carbon-free renewable energy. Today, fossil power (natural gas and coal) still supplies most of the energy used by data centers worldwide \citep{iea2025}. As such, the latter scenario has been strongly advocated: the United Nations recently urged AI developers to run data centers on 100\% renewable energy by 2030 \citep{reuters2025}. Understanding how the energy mix powering AI will evolve, and whether this ``100\% renewable by 2030'' vision can be achieved, remains an open question. Prevailing views fall into two camps.

A first view frames AI's energy footprint as a capacity problem: renewables are scaling, so most incremental data-center load will likely be met by new renewable supply. The International Energy Agency projects that renewables will meet about half of the growth in data-center electricity demand, with natural gas and coal supplying much of the remainder, and the split varying substantially across scenarios \citep{iea2025}. Some go further, arguing that AI-driven demand could itself accelerate clean investment and hasten substitution away from fossil generation \citep{fransoo2026navigating,spglobal2025}, a view echoed in the emerging ``power couple'' framing of AI growth and renewable investment \citep{wef2025powercouple,rmi2025powercouple,iea2023powercouple}. The underlying logic is that renewables' widening cost advantage makes clean procurement increasingly attractive \citep{irena2025}, and AI-enabled grid optimization can ease integration frictions \citep{wef2025}. In this account, emissions hinge on the pace of clean-capacity deployment.

A second view warns the opposite:  rapid growth in data-center demand can \emph{entrench} fossil generation. Because many data centers are concentrated in regions with carbon-intensive grids, utilities facing abrupt load increases often meet reliability needs by running existing coal and gas plants harder and postponing planned retirements \citep{eesi2025}. Industry reporting also documents cases in which utilities repower shuttered fossil assets \citep{tbj2025} or commission new gas-fired capacity specifically to serve data-center loads \citep{utilitydive2025}. In this account, clean-capacity expansion alone is not enough to reduce emissions. Instead, if clean capacity does not scale sufficiently in step with compute, AI growth can create a carbon \emph{lock-in} that slows the renewable transition.

Reconciling these views is difficult because AI development and the energy transition coevolve %
\citep{cohen2026omforum,netessineShunko2026InvisibleBackbone}.
 Notably, the pace of the energy transition influences economic incentives for AI development, while AI itself can drive or impede that transition through its substantial power demand and potential grid impacts. Capturing this interplay and analyzing the resulting equilibrium are essential to understanding how AI and energy shape each other. Existing analyses typically treat developments in the AI and energy sectors independently, focusing on one side (e.g., AI innovation or clean-energy investment) while holding the other side exogenous.

This paper offers an initial step toward filling that gap. We develop a model in which renewable capacity and AI capability are set by separate actors whose choices feed back on each other: clean-energy availability affects the cost of scaling AI, and the resulting scale changes the economic and environmental value of additional clean capacity. To capture this interaction tractably, the baseline model puts a welfare-maximizing policymaker on the clean-energy side and a profit-maximizing developer on the AI side that selects a capability level (for instance, a frontier, computation-intensive AI versus a more modest, utility-grade AI). The policymaker's capacity choice stands in for the policy and investment machinery that determines clean-energy availability; an extension explicitly introduces a utility that responds to policy incentives.%

Our model captures three features of the AI context. First, the compute and energy a given capability requires depend on the prevailing scaling regime, which shapes how much renewable investment lowers operating costs. Across regimes, the capability--compute relationship can be close to proportional over the relevant range or exhibit sharply diminishing marginal gains from additional compute. We capture this variation through a scaling parameter. \looseness=-1

Second, AI markets can exhibit supermodular returns to capability: incremental breakthroughs unlock new uses and generate disproportionately large value. As the technology matures and key applications saturate, these returns may become submodular. We represent these differences through the elasticity of willingness-to-pay. That elasticity sets how sharply returns curve, but that curvature alone does not determine whether capability reaches frontier scale; what matters is how fast willingness-to-pay rises relative to the energy needed to raise capability. \looseness=-1 

Third, the model allows AI to generate climate-related \emph{benefits}, particularly through adaptation and resilience.%
\footnote{AI may also help reduce emissions (mitigation), for example, by facilitating renewable integration or optimizing energy efficiency \citep{Choi2024}. However, such effects are more uncertain (due to potential rebound effects) and slower to materialize, whereas adaptation benefits are more immediate and observable in practice. We focus on adaptation benefits as the primary social value of AI in the climate context.}
Accordingly, we capture a potential trade-off: although capability is energy-intensive, its social value could rise with climate damage because it also aids climate adaptation, even as its energy use can worsen that damage.

Solving the model reveals two distinct ways in which AI growth can interact with the energy transition. When the market value of capability grows at least as fast as its energy requirements, as may be true near today's frontier, developers have strong incentives to pursue frontier capability even when doing so requires fossil power. Renewable investment can therefore support further AI growth without fully displacing fossil generation. This creates an ``adaptation trap'': as climate damages worsen, AI becomes more valuable for adaptation, strengthening incentives to sustain frontier capability even while some of its electricity continues to come from fossil sources. The equilibrium combines high AI capability with persistent emissions.

When the market value of capability grows more slowly than its energy requirements, developers become more sensitive to energy costs. Expanding renewable capacity can then lower the cost of additional AI capability while reducing the emissions associated with it. As climate damages worsen and AI becomes more valuable for adaptation, the policymaker has a stronger incentive to expand clean capacity. Beyond a threshold, renewable capacity can grow enough to meet the full electricity needs of AI development. We call this an ``adaptation pathway'': greater climate stress encourages clean-energy investment, which in turn supports further AI development without additional fossil emissions.

We complement these analytical results with a calibrated case study that captures observed differences in electricity markets, renewable costs, and AI demand growth. The numerical examples show that both the adaptation trap and the adaptation pathway arise under empirically plausible conditions. They also show that powering AI entirely with renewable energy may remain difficult in some settings unless the cost of expanding renewable capacity falls substantially and climate policy remains sufficiently strong.

We also examine two extensions. First, AI may itself lower renewable-energy costs, for example by improving grid integration or accelerating innovation. Such cost reductions can have an unexpected effect: unless they are large enough to support full decarbonization, they can make frontier AI more attractive while leaving some fossil use in place, thereby worsening the adaptation trap. Second, competition among AI developers can promote cleaner development when market potential is moderate. Competition reduces profit margins, making firms more sensitive to energy costs and increasing their incentives to use renewable power, which in turn encourages renewable investment. When market potential is sufficiently strong, however, margins remain large enough to sustain frontier scaling with fossil power, and both developers can fall into the adaptation trap.

The remainder of the paper is organized as follows. \cref{sec:literature} discusses related literature. \cref{sec:conceputal_framework} introduces the modeling framework. \cref{sec:developer-decision,sec:policymaker-decision} analyze the AI developer's and policymaker's decisions, respectively, deriving the main equilibrium results. \cref{sec:case-study} presents the case study, and \cref{sec:extensions} considers the extensions. %
\cref{sec:conclusions} concludes with implications for policy and for future research on the AI--climate nexus.

\section{Literature}
\label{sec:literature}

The AI boom has made electricity a binding constraint on innovation, so energy policy increasingly functions as technology policy. Understanding the interaction between AI and energy means studying technological change, clean-capacity investment, and climate damages together. We organize the related work around three streams.

\subsection{Environmental Policy and Technology Innovation}

First, our paper relates to work on directed innovation and carbon lock-in. In the directed technical change tradition, policy shifts private returns to steer invention; in a canonical model, \citet{acemoglu2012environment} show that well-designed climate policy can redirect innovation toward clean technologies and sustain the transition even after policy is relaxed. Empirically, higher energy prices and taxes reallocate inventive activity toward cleaner technologies \citep{aghion2016auto,popp2002}. A parallel tradition considers why cost-effective energy-saving technologies often diffuse slowly, the so-called ``energy paradox,'' and traces it to a combination of pollution and innovation-and-adoption market failures \citep{jaffe_stavins_1994,jaffe_newell_stavins_2005}, whereas a long-running debate considers whether well-designed environmental regulation can itself stimulate offsetting innovation \citep{porter_vanderlinde_1995}. Complementary operations management work makes the policy--diffusion margin explicit at the firm and market level, from early models of green design responses to environmental requirements \citep{chen2001design} to more recent analyses of how subsidy design, strategic supply, and regulatory uncertainty shape adoption and investment \citep{cohen2019commitment,cohen2016demanduncertainty,wang2021green}. We share this literature's interest in how environmental policy shapes technology, but the policy margin differs. Rather than steering invention between clean and dirty alternatives, clean-energy investment here sets the scale at which a developer deploys a single energy-intensive, adaptation-providing technology, and that equilibrium scale is what we study. \looseness=-1

The lock-in literature starts from a different diagnosis. Increasing returns, network effects, and institutional complementarity can entrench an inferior regime long after superior alternatives exist \citep{arthur1989}. In energy systems, carbon-intensive technologies are embedded in a dense web of infrastructure, regulation, skills, and sunk investments that together form a techno-institutional complex resistant to change \citep{unruh2000,unruh2002}. These perspectives raise a central question for the AI context: if policy can redirect innovation, why does carbon lock-in remain durable, and if lock-in dominates, why do transitions occur at all? We address this question by modeling the joint determination of the innovation margin and the capacity margin and showing that %
lock-in can arise endogenously from the equilibrium interaction between private scaling incentives and public clean-capacity choices under certain conditions. Counterintuitive incentive effects of this kind recur in the economics of AI deployment more broadly: more reliable AI, for instance, can paradoxically make it costlier to sustain human oversight \citep{bastani_cachon_2025}, and the most capable machine need not be the one a firm should deploy once it must also motivate human effort \citep{guan_qi_wang_2025}; relatedly, when a human will only partially follow algorithmic advice, the best recommendation departs from the best decision \citep{grand_clement_pauphilet_2026}.

\subsection{Clean Capacity Investment and Carbon Footprint of AI}

Second, a literature in operations and energy economics studies clean-capacity investment under uncertainty, accounting for intermittency and valuing operational flexibility, to meet load at acceptable cost and reliability. Much of this work treats demand as given, or as responding to prices through reduced-form elasticities and demand-response programs. Recent contributions endogenize strategic behavior on the supply side: \citet{aflaki2017} analyze renewable investment under intermittency and strategic behavior, and \citet{kok2020} show how operational flexibility reshapes the optimal mix of renewable and conventional generation. Related work studies joint choices over renewables, storage, and load control under tight resource constraints, yielding micro-foundations for data-center and ``behind-the-meter'' procurement and showing how clean investment depends on operational substitutability across these margins \citep{kaps2023offgrid}. Closest to our two-source energy abstraction, \citet{peng_wu_souza_2024} study joint operations and investment across renewable, flexible (natural-gas), and storage capacities and characterize when these resources act as substitutes or complements.  A closely related strand emphasizes contracting and market design: \citet{gao_sunar_birge_2025} analyze renewable power purchase agreements (PPAs) and show how contract structure and timing affect green-capacity expansion, whereas work on corporate renewable targets highlights how commitment mechanisms and investment timing interact with uncertainty \citep{sunar_birge_2019,trivella_mohseni_taheri_nadarajah_2023}. At the decentralized margin, \citet{gao_alshehri_birge_2024} show that aggregating distributed energy resources can improve efficiency while also creating market power concerns, and operational constraints such as curtailment imply that installed capacity need not translate one-for-one into clean output \citep{wu_kapuscinski_2013}. We differ from this stream in what is treated as endogenous: it optimizes a supply portfolio against given or price-responsive demand, whereas we make AI load endogenous to clean-energy availability; investment in clean capacity changes the economics of AI scaling, which in turn changes equilibrium electricity demand.

In parallel, an emerging literature quantifies the energy use and carbon footprint of the digital economy, especially large-scale AI. \citet{strubell-etal-2019-energy} highlight that frontier model development can be energy-intensive, motivating emissions accounting and efficiency efforts. Policy and industry analyses warn that data-center electricity demand could grow rapidly with AI workloads \citep{wef2025,eesi2025} and document local grid stress from clustered data-center expansion \citep{utilitydive2025}. Firms have responded with decarbonization commitments and record corporate renewable procurement \citep{bnef2024corporate,google2020environmental,microsoft2020carbonnegative}. A common limitation is that these analyses rarely model how developers adjust capability, and thus compute intensity, when energy becomes cheaper or cleaner. Capacity-planning and procurement models provide detail on supply, contracts, and market design but typically treat demand as technologically passive; the AI-footprint literature treats compute growth as an outcome rather than a strategic choice. Neither captures that clean-capacity investment can raise equilibrium compute demand by lowering the effective marginal cost of electricity, creating scope for rebound-type effects \citep{Gillingham2016}. We model this feedback explicitly by endogenizing both clean-capacity investment and AI scaling. \looseness=-1

\subsection{Adaptation and Mitigation}

Third, a climate-policy literature examines whether adaptation can expand without weakening mitigation. In the benchmark framing, adaptation and mitigation are substitutes, so cheaper adaptation lowers the marginal benefit of abatement. Integrated assessment work that incorporates adaptation formalizes this channel and shows how the timing of adaptation shapes optimal mitigation \citep{debruin2009}, with related arguments in surveys emphasizing uncertainty and learning \citep{ingham2007}. At the same time, syntheses stress that climate policy is a portfolio problem: adaptation is unavoidable in the near term, but without mitigation, warming can outstrip adaptive capacity \citep{ipcc2007wg2ch18}. This tension underlies concerns about an ``adaptation moral hazard.''

Beyond the benchmark, crowd-out is not inevitable. Adaptation can shift the distribution and timing of losses and the political constraints on mitigation, allowing complementarity \citep{ingham2013}. What this debate largely abstracts from is adaptation delivered by a technology with steeply increasing returns to scale. In particular, AI-based adaptation is compute-intensive and scalable, with complementarities between capability and deployment; with supermodular returns, small advantages compound, and scale races become self-reinforcing \citep{bulow1985,parker2005}. In AI markets specifically, a data feedback loop known as the ``AI flywheel,'' in which wider adoption generates more data that further improves capability, reinforces these dynamics \citep{gurkan_devericourt_2022}. Embedding this scaling logic into the adaptation--mitigation problem changes the relevant margin: worsening climate damages can raise the return to capability, translating adaptation demand into faster compute growth and, unless marginal electricity is carbon-free, higher emissions. Our model formalizes this interaction and shows that  %
when a technology with steeply increasing returns delivers adaptation, how its energy interface is governed shapes whether adaptation and mitigation are complementary or in tension.

\section{Model}
\label{sec:conceputal_framework}

We develop a parsimonious model that isolates the core incentives linking AI development and renewable investment. The model endogenizes both AI scale and the clean capacity available to support it; for tractability, we represent their interaction as a sequential game. In a single planning period, two actors move sequentially. A welfare-maximizing policymaker chooses renewable capacity $y$ intended to power AI development. A profit-maximizing AI developer then chooses an AI capability level $x\in[0,1]$.

The capacity $y$ should be read as ``dedicated'' clean supply available to the developer, for example, via long-term power purchase agreements (PPAs) or on-site generation, rather than undifferentiated grid supply. This focus matches how AI developers often secure clean electricity and clarifies the margin at which policy and corporate procurement interact.

\subsection{Capability, Energy Demand, and Emissions}
We consider a single period, a technological phase in which a prevailing paradigm governs the mapping from compute to capability. For example, in one phase, pretraining scale drives progress: larger datasets, more parameters, more training compute. In another, incremental compute is shifted toward inference (``test-time'' scaling). The one-period structure provides a tractable way to study how the equilibrium between clean capacity and AI capability depends on the economics of scaling within a phase.

Within the phase, the developer chooses an AI capability level $x$ that indexes performance along the scaling frontier. We normalize the maximum attainable performance in the phase to $1$, reflecting that improvements eventually plateau. The developer chooses $x\in[0,1]$ to maximize profits, trading off market value against costs. %

Compute is energy-intensive. Motivated by empirical ``scaling laws'' that relate model performance to training compute through approximate power laws \citep{hestness2017deep,hoffmann2022train,kaplan2020scaling}, we assume that achieving capability $x$ requires compute equal to $x^{1+\alpha}$, where $\alpha\ge 0$ governs how rapidly compute requirements rise with performance. If each unit of compute consumes $k$ units of energy, the developer's energy demand is $k\,x^{1+\alpha}$.

Developers procure electricity from the grid and from dedicated renewables (PPAs or on-site generation).\footnote{In practice, developers' alternative to grid power is likely a portfolio of energy resources, including dedicated renewable generation, efficient natural-gas technologies such as on-site combined cooling, heat, and power systems, and energy storage. We abstract from this mix and focus on the key economic and environmental attributes it typically delivers: lower generation costs and emissions than grid power, but at the expense of substantial upfront investment.} %
In practice, grids are mixed and remain dominated by fossil generation. 
In the United States, natural gas and coal together supply roughly three-fifths of electricity generation \citep{eia2025epm}; in China, coal generates the majority of electricity \citep{eia2025china}. Renewable PPAs, by contrast, match most energy demand with renewable generation under varying matching requirements. %
In the model, we abstract the two sources as ``fossil'' and ``renewable'' to contrast carbon-intensive and carbon-free supply. %

The developer incurs unit costs of $c_f$ and $c_r$ for energy from the fossil and renewable sources, respectively. In the main model, we interpret these parameters as unit generation costs, as in cases in which developers use on-site energy production. Developers procuring electricity from energy operators through PPAs or from utilities at prevailing electricity tariffs likely pay more than the underlying generation costs. The margins accrue to the energy operators or utilities and drive the strategic behavior of these entities. We abstract from this strategic layer and focus on the underlying generation cost in the main model. We refer readers to \cref{subsec:decentralized} for an extended model that explicitly incorporates a utility, its strategic decision, and the associated cash flows.

We take $c_r<c_f$ to reflect the general cost advantage of renewable sources in energy generation. %
For the time period considered, we treat renewables as capacity-constrained, whereas fossil generation can be ramped to cover residual demand. This captures the practical asymmetry that clean expansion requires lead time (siting, interconnection, transmission, storage), whereas fossil capacity in the grid can respond relatively quickly to load growth, including from data centers \citep{utilitydive2025,eesi2025}. For example, in Virginia, which has the world's highest concentration of data centers, utilities have pursued aggressive generation and transmission expansion plans, including substantial fossil capacity additions, to keep pace with AI demand growth.\footnote{We acknowledge that grid capacity is constrained in practice. However, such constraints are typically reflected in interconnection queues for new data centers. Once connected, data centers generally rely on the utility's obligation to serve, under which residual demand is met through grid power.  }

Because $c_r<c_f$, a profit-maximizing developer uses all available renewables first and relies on fossil energy only for demand above renewable capacity. Accordingly, 
given the developer's capability choice $x$ and the renewable capacity $y$,\footnote{The model can be extended to include an existing renewable capacity level $y_0$ in addition to the investment $y$; all results remain robust.} the share of the developer's energy demand met by renewables is
\begin{equation} \label{eq:beta_x_y}
\beta(x,y) \;=\; \min\left\{\frac{y}{kx^{1+\alpha}},\,1\right\}.
\end{equation}
If $kx^{1+\alpha}\le y$, renewables cover all energy use and $\beta(x,y)=1$; otherwise, $\beta(x,y)<1$.

A more rigorous formulation of \cref{eq:beta_x_y} would replace renewable capacity $y$ in the numerator with the quantity of renewable energy that can be effectively deployed to satisfy the developer's electricity demand. The distinction matters because leading renewable sources, including solar and wind, are intermittent: the same nameplate capacity may yield different levels of usable renewable electricity depending on natural conditions and the effectiveness of firming measures, such as energy storage and data centers' ability to shift or flexibly schedule workloads. \Cref{subsec:utilization} examines this refinement explicitly: the developer's problem retains its structure under intermittent availability, whereas both scaling regimes tilt toward more carbon-intensive outcomes. \looseness=-1

Market willingness-to-pay for capability $x$ is $\theta x^\lambda$, where $\theta>0$ denotes the valuation at the capability frontier, $x=1$, and $\lambda>0$ governs the elasticity of value with respect to capability. When $\lambda>1$, private returns to capability are convex: incremental improvements can unlock new tasks and markets, so valuation rises more than proportionally with $x$. When $0<\lambda<1$, by contrast, marginal gains become progressively less valuable. The monopoly benchmark can therefore be interpreted as a reduced-form representation of a regime with concentrated rents and strong incentives to race toward the frontier; \cref{subsec:competition} relaxes this assumption by introducing competition.

Achieving a given capability entails multiple cost components (hardware, facilities, labor, and energy). We assume energy sources differ only in their unit energy cost, normalizing other costs away.\footnote{In practice, the energy source can affect operations (e.g., workload shifting). We abstract from these operational margins to focus on investment and scaling incentives. A common capability-dependent non-energy cost would shift the developer's activation thresholds by raising the effective unit cost of compute, but it enters both energy branches symmetrically and leaves the regime comparison between $\lambda$ and $1+\alpha$ unchanged.} Equivalently, $c_r$ and $c_f$ are all-in unit compute costs under renewable versus fossil-based power.

The developer's profit is
\begin{equation}
\Pi(x,y) \;=\; \theta x^\lambda
-\Big(\beta(x,y)c_r+\big(1-\beta(x,y)\big)c_f\Big)\,k\,x^{1+\alpha}.
\end{equation}
The optimal capability choice conditional on $y$ is $x^*(y)\doteq \arg\max_{x\in[0,1]} \Pi(x,y)$.

\subsection{Renewable Capacity and the Policy Objective}

In practice, decisions regarding renewable-capacity investment are distributed across governments, utilities, and private firms, including AI companies themselves. Nevertheless, public policy heavily shapes investment outcomes. Governments can directly invest in clean generation, transmission, and storage infrastructure; utility investments in renewable projects are often subject to regulatory approval and procurement requirements such as Renewable Portfolio Standards; technology companies' investments in clean firm power can be facilitated by policies such as Clean Transition Tariffs; and third-party renewable developers frequently rely on financial incentives, including tax credits and investment subsidies.

Collectively, these mechanisms imply a mapping from policy choices to renewable investment outcomes. We therefore interpret $y$ as the clean capacity made available to AI through this broader investment and policy environment. The baseline model collapses these institutions into a welfare-maximizing actor that directly chooses $y$; this collapse isolates the feedback between clean-energy availability and AI scaling. This single-instrument abstraction allows us to focus on the fundamental interaction between policy-induced clean-energy investment and market-driven AI demand growth, suppressing implementation frictions that may arise in translating policy instruments into realized investment outcomes. We formalize the policy-to-investment process in \cref{subsec:decentralized}.

The relevant decision variable is the period-specific capacity $y$ that is available to power the phase of AI development under study. Although much energy infrastructure is long-lived, climate and industrial policy are frequently revised, and AI itself is a salient driver of such revisions. The one-period formulation can therefore be understood as capturing an investment decision within a phase, %
and examining the clean-capacity response to prospective AI load within a given scaling regime.

A central policy concern is emissions from fossil-powered compute and the resulting climate damages. Let $d_0\geq 0$ denote the baseline emissions stock in the period considered, absent AI. 
This can be interpreted as the atmospheric carbon stock above a given climate stability threshold.
Let $e_f>0$ be the carbon emissions per unit of fossil-generated electricity, and treat renewables as zero-carbon. Given $(x,y)$, emissions attributable to AI development are
\begin{equation}
E(x,y) \;=\; e_f\big(1-\beta(x,y)\big)\,k\,x^{1+\alpha},
\end{equation}
so total emissions are $d_0+E(x,y)$. Both $d_0$ and $E$ are measured in the same unit (e.g., million metric tons of CO$_2$e).

At the same time, higher AI capability can reduce the realized damages from a given emissions stock by improving forecasting, preparedness, and other adaptation services. We capture this adaptation channel through a reduced-form formulation in which capability $x$ scales down climate damages by a factor $1-bx$, where $b\in (0,1)$ measures the effectiveness of AI-enabled adaptation. %
The resulting climate-damage term is
\begin{equation}
D(x,y) \;=\; (1-bx)\,\big(d_0+E(x,y)\big).
\end{equation}

Expanding renewable capacity is costly and increasingly so at the margin: prime sites are developed first, and further scale requires transmission, interconnection, and balancing resources. We therefore model investment cost as an increasing, convex function $V(y)$ with $V(0)=0$. In an extension, we allow $V$ to depend on capability, $V(x,y)$, to capture the possibility that higher AI capability lowers the effective cost of integrating renewables, for example, through improved grid management and planning (see \cref{subsec:reductions-renewable-costs}).

The policymaker chooses $y$ to maximize social welfare, anticipating the developer's best response $x^*(y)$. Welfare includes economic surplus generated by AI (which may exceed the developer's private profit when there are spillovers), subtracts climate damages net of adaptation benefits, and subtracts the investment cost of renewable capacity. A compact representation is
\begin{equation}
W(y)
\;=\;
\underbrace{\eta\,\Pi\big(x^*(y),y\big)}_{\text{economic surplus}}
\;-\;
\underbrace{\xi\,D\big(x^*(y),y\big)}_{\text{climate damages net of adaptation}}
\;-\;
\underbrace{V(y)}_{\text{renewable-investment cost}},
\end{equation}
where $\eta>1$ captures spillovers from AI profitability to broader economic value creation and $\xi>0$ scales the policymaker's weight on climate damages relative to economic surplus.

It is without loss of generality to restrict the policy choice to $y\in[0,k]$ so that dedicated renewable energy does not exceed the maximum demand $k$. The optimal renewable capacity is
\begin{equation}
y^* \;\doteq\; \arg\max_{y\in[0,k]} W(y).
\end{equation}

\section{Developer Capability Choice}
\label{sec:developer-decision}

We solve by backward induction, starting from the AI developer's best response to renewable capacity $y$. The developer chooses capability $x\in[0,1]$ to maximize profits, trading off market value $\theta x^\lambda$ against the energy cost of supplying $k x^{1+\alpha}$ units of electricity at an endogenous unit price determined by the clean share $\beta(x,y)$. A key force is how willingness-to-pay scales with capability relative to how energy requirements scale, captured by $\lambda$ versus $1+\alpha$. Intuitively, when $\lambda \ge 1+\alpha$, returns keep pace with energy needs, and the developer prefers higher capability; when $\lambda < 1+\alpha$, energy costs become increasingly binding at high $x$.

We begin with $\lambda \ge 1+\alpha$, an environment in which the marginal private payoff from capability does not taper faster than the marginal energy requirement. This aligns with convex valuations in frontier markets and scaling regimes in which additional compute remains effective. We refer to this as \emph{market-led scaling}. When multiple capability choices yield the same profit, we assume the developer selects the larger $x$.

\begin{proposition}
\label{lem:1}
    In the market-led scaling scenario in which $\lambda\geq 1+\alpha$, there exist two thresholds $y_1\leq y_2$ such that for any given $y\in [0,k]$, %
    the optimal AI capability choice $x^*$ equals
    \begin{equation}
        x^*(y)=\begin{cases}
            0 \quad \forall y< y_1\\
            \left(\frac{y}{k}\right)^{\frac{1}{1+\alpha}} \quad \forall y\in [y_1,y_2)\\
            1 \quad \forall y\geq y_2
        \end{cases}
    \end{equation}
    Accordingly, the emissions from AI development $E(x^*(y),y)=0$ if $y<y_2$. Otherwise, $E(x^*(y),y)>0$ unless $y=k$.
\end{proposition}
\emph{Proof.} All proofs of results stated in the main text are relegated to \cref{ec:proofs}.

\cref{lem:1} shows that the link between renewable capacity and fossil reliance can be discontinuous. The developer's best response has three regions. When $y<y_1$, the developer is inactive ($x^*(y)=0$): even modest capability would require enough fossil power that energy costs dominate market value.

When $y\in[y_1,y_2)$, the developer produces a positive-capability model but stays fully within renewables. In particular,
\[
x^*(y)=\left(\frac{y}{k}\right)^{\frac{1}{1+\alpha}}
\]
exactly exhausts clean capacity and yields $E(x^*(y),y)=0$. Renewables bind, so capability tracks the clean-energy constraint exactly. We call this interval the ``coupling zone.''

Once $y$ crosses $y_2$, the developer jumps to the frontier choice $x^*(y)=1$. Unless $y\ge k$, residual demand is met by fossil generation, and emissions become strictly positive. Thus, additional renewables can shift the equilibrium from a clean, capacity-constrained regime to a frontier regime in which fossil power is used at the margin.

The mechanism is straightforward. Under $\lambda\ge 1+\alpha$, the marginal market payoff from capability remains large relative to the implied increase in energy requirements, so the developer is drawn toward $x=1$. When $y$ is small, reaching high capability entails paying $c_f$ on a large fossil share, which deters development or keeps it tightly constrained by clean capacity. As $y$ rises, the developer can scale while remaining carbon-free. Beyond $y_2$, the remaining fossil cost at $x=1$ is no longer decisive, and the developer absorbs it rather than forgo frontier capability. For $y\ge y_2$, capability no longer depends on renewable availability; we refer to this region as the ``decoupling zone.''

We next consider the case in which compute requirements rise faster than willingness-to-pay as capability increases, that is, $\lambda<1+\alpha$. We refer to this as the \emph{resource-led scaling} scenario. Define
\[
f(c)\;\doteq\;\left(\frac{\theta\lambda}{(\alpha+1)\,c\,k}\right)^{\frac{1}{1+\alpha-\lambda}},
\]
which is decreasing in $c$.

\begin{proposition}
\label{lem:2}
In the resource-led scaling scenario in which $\lambda<1+\alpha$, for any $y\in[0,k]$, the optimal capability choice is
\begin{equation}
x^*(y)=
\begin{cases}
\min\{f(c_f),1\}, & \text{if } y<k\,f(c_f)^{1+\alpha},\\[4pt]
\left(\frac{y}{k}\right)^{\frac{1}{1+\alpha}}, & \text{if } y\in\bigl[k\,f(c_f)^{1+\alpha},\,k\,f(c_r)^{1+\alpha}\bigr),\\[6pt]
\min\{f(c_r),1\}, & \text{if } y\ge k\,f(c_r)^{1+\alpha}.
\end{cases}
\end{equation}
Accordingly, $E(x^*(y),y)=0$ for $y\ge k\,f(c_f)^{1+\alpha}$; otherwise, $E(x^*(y),y)>0$ unless $y=k$.
\end{proposition}

\cref{lem:2} implies a different pattern from the market-led case. Capability remains nondecreasing in $y$, but the best response is typically interior and varies continuously: when $\lambda<1+\alpha$, the marginal energy requirement grows faster than the marginal market payoff, so the developer becomes increasingly sensitive to the unit energy price at high $x$. Expanding renewables therefore mainly works by lowering the effective cost of capability, inducing smooth adjustments in $x$ rather than a jump to the frontier.

The emissions implications also reverse. For low $y$, the developer behaves as if facing the fossil cost $c_f$ and chooses $x^*(y)=\min\{f(c_f),1\}$, which is independent of $y$ and generally entails fossil use, a decoupling zone. Once $y$ reaches $k\,f(c_f)^{1+\alpha}$, renewables suffice to support the developer's preferred scale without fossil generation, so emissions fall to zero. Over the intermediate range, higher $y$ lowers the effective unit cost toward $c_r$ and can raise capability on a carbon-free basis; for sufficiently large $y$, the unit cost is pinned down by $c_r$, and $x^*(y)$ again becomes independent of $y$.

In sum, the response to clean capacity depends on whether scaling is market-led ($\lambda\ge 1+\alpha$) or resource-led ($\lambda<1+\alpha$). In the former, higher $y$ can unlock frontier capability with residual fossil use at the margin; in the latter, higher $y$ primarily substitutes renewables for fossil supply and supports interior capability choices. The next section shows how this distinction shapes the policymaker's renewable investment decision.

\section{Renewable Capacity in Equilibrium}
\label{sec:policymaker-decision}

Policy and industry discussions increasingly call for powering frontier AI with carbon-free electricity, often presuming that renewable expansion will naturally follow AI's rising load \citep{reuters2025}. However, whether this is true and whether it delivers decarbonization depend on strategic responses. In our framework, renewable capacity $y$ affects welfare through two margins: it raises the clean share $\beta(x,y)$ at a given capability level, and it shifts equilibrium capability through the developer's best response $x^*(y)$ (characterized in \cref{sec:developer-decision}). 
These two margins correspond to a \emph{composition} effect and a \emph{scale} effect: holding $x$ fixed, higher $y$ displaces fossil generation, but higher $y$ can also increase total energy demand by raising $x^*(y)$.
The policymaker chooses $y$ anticipating $x^*(y)$, and the relevant outcome is the induced equilibrium pair $(x^*(y^*),y^*)$ and its implied emissions and damages.
Accordingly, whether clean investment reduces emissions or instead accelerates them depends on which effect dominates in equilibrium and, ultimately, on the prevailing AI scaling regime. We analyze these effects in this section.

\subsection{Market-Led Scaling: The Adaptation Trap}
We begin with market-led scaling, $\lambda\ge 1+\alpha$, in which willingness-to-pay rises at least as fast as energy requirements as capability increases. 
As \cref{lem:1} illustrates, in this regime, %
clean capacity plays an enabling role %
and can move the developer's capability choice to the frontier even when full decarbonization is not achieved. We show that such a decoupling between frontier capability and full decarbonization can arise at the optimal level of renewable investment. \looseness=-1 

\begin{proposition}
\label{prop:1}
    In the market-led scaling scenario in which $\lambda\geq 1+\alpha$, at the equilibrium renewable capacity $y^*$,
    \begin{enumerate}
        \item[(i)] if either (a) $\theta\leq \max\{\frac{1+\alpha}{\lambda }\cdot c_fk, c_rk\}$, or (b) $V'(k)\leq \eta (c_f-c_r)+(1-b) e_f \xi$, then $E(x^*(y^*),y^*)=0$ holds $\forall d_0\geq 0$;
        \item[(ii)] otherwise, there exists a threshold $\bar{d}_m$ such that $E(x^*(y^*),y^*)>0$ if and only if $d_0\geq \bar{d}_m$. Furthermore, $y^*$ remains unchanged $\forall d_0\geq \bar{d}_m$.
    \end{enumerate}
\end{proposition}

Part (i) describes two cases in which the equilibrium is carbon-free even under market-led scaling. Under (a), frontier capability is not sufficiently valuable, relative to energy costs, to justify the discrete jump to $x=1$ in \cref{lem:1}. The developer instead chooses an interior capability level that can be served by renewables, or remains inactive, so $E(x^*(y^*),y^*)=0$. Under (b), full decarbonization is socially efficient: the marginal cost of expanding clean capacity to $k$ is below its marginal benefit, which combines the private cost saving $\eta(c_f-c_r)$ with avoided climate damages $(1-b)e_f\xi$. The policymaker therefore sets $y^*=k$, delivering frontier capability on carbon-free power.

Part (ii) characterizes the remaining region, in which frontier incentives are strong and the last units of clean capacity are expensive. Here climate deterioration can tilt the equilibrium toward carbon-intensive AI. Because damages take the form $(1-bx)(d_0+E(x,y))$, a higher baseline emissions stock $d_0$ increases the marginal value of capability through adaptation: holding emissions fixed, a higher $x$ reduces the damage burden more. In the market-led regime, however, once $y$ is high enough to make scaling feasible, the developer's choice is largely pinned down by market incentives, not by marginal energy costs. Once clean capacity suffices to support frontier scaling, further renewable investment no longer raises capability; the remaining gain, displacing fossil electricity, may not justify the last units of capacity. The threshold $\bar d_m$ formalizes the trade-off: beyond it, the welfare gain from inducing frontier capability (the scale effect) is already maximized, whereas the increase in investment cost outweighs the welfare loss from residual marginal emissions (the composition effect). Consequently, renewable investment stops short of the decarbonization level $y=k$, at or moderately above the level that triggers $x^*(y)=1$.

This logic delivers an \emph{adaptation trap}. Partial clean investment is attractive because it enables frontier capability and hence adaptation benefits. But when frontier scaling is powered at the margin by fossil electricity, equilibrium emissions remain positive, and the worse the baseline, the greater the value of adaptation and hence the stronger the incentive to keep enabling frontier scaling while tolerating residual fossil use. To summarize, carbon intensity persists under market-led scaling because adaptation benefits from frontier capability (the scale effect) outweigh the benefits of full decarbonization (the composition effect).

\subsection{Resource-Led Scaling: The Adaptation Pathway}
We next consider resource-led scaling, $\lambda<1+\alpha$, in which energy requirements rise faster with capability than willingness-to-pay. In this regime, the developer's capability choice becomes increasingly sensitive to the unit cost of energy. Clean-energy availability therefore has greater influence over AI scaling: by lowering the marginal cost of electricity serving AI load, it changes the private return to scaling and can shift $x^*(y)$, as characterized in \cref{lem:2}. This stronger coupling between energy costs and capability choice yields a qualitatively different outcome from that under market-led scaling. \looseness=-1

\begin{proposition}
\label{prop:2}
    In the resource-led scaling scenario in which $\lambda<1+\alpha$, at the equilibrium renewable capacity $y^*$,
    \begin{enumerate}
    \item[(i)] if $\theta\geq \frac{1+\alpha}{\lambda}\cdot c_fk$ and $V'(k)> \eta(c_f-c_r)+(1-b) e_f \xi$, then $E(x^*(y^*),y^*)>0$ holds $\forall d_0\geq 0$;
    \item[(ii)] otherwise, there exists a threshold $\bar{d}_c$ such that $E(x^*(y^*),y^*)=0$ if and only if $d_0\geq \bar{d}_c$.
    \end{enumerate}
\end{proposition}

Part (i) identifies the corner case in which carbon-intensive AI persists: market incentives to scale remain strong, whereas the marginal cost of fully covering AI load with renewables is high relative to its benefits. In this region, the policymaker rationally stops short of full decarbonization, and some of the electricity used by AI continues to come from fossil generation.

Outside this region, climate deterioration can push the equilibrium toward carbon-free AI. The first step parallels the market-led case: a higher baseline emissions stock $d_0$ raises the social value of adaptation, so the policymaker induces higher capability. The second step is where resource-led scaling changes the conclusion. Because the developer is energy-cost sensitive, higher capability %
requires a meaningful reduction in the effective marginal cost of scaling. 
This potentially aligns the welfare gains from adaptation (through the scale effect) and mitigation (through the composition effect), because investing in renewables accomplishes both: it substitutes away from emissions-intensive supply and lowers the energy cost faced by the developer over the relevant range. %
As a result, when $d_0$ is large enough, the policymaker finds it optimal to invest sufficiently in clean capacity so that the induced increase in $x^*(y)$ occurs on carbon-free marginal electricity, yielding $E(x^*(y^*),y^*)=0$. The threshold $\bar d_c$ captures when adaptation value is high enough to justify this shift.

This mechanism yields an \emph{adaptation pathway}. Under resource-led scaling, worsening climate conditions increase the value of capability, and that increase translates into stronger incentives to expand clean capacity, because clean capacity is precisely what makes additional capability privately viable. In this regime, the feedback runs in a decarbonizing direction: climate stress strengthens clean investment, and clean investment supports capability growth on carbon-free power.

The adaptation pathway mechanism can be undermined by intermittency. Notably, an increase in nameplate renewable capacity could increase emissions under high output variability when scaling is not resource-led enough (i.e., when $1+\alpha$ is not sufficiently higher than $\lambda$). Consequently, crossing the threshold $\bar d_c$ may shift the equilibrium to a more carbon-intensive region, thereby triggering an adaptation trap. We refer readers to \cref{subsec:utilization} for details. \looseness=-1

\subsection{Policy Design Implications}
The two regimes yield policy lessons that are difficult to see in analyses that hold either AI scaling or clean-capacity expansion fixed. In particular, ``net-zero AI'' is an equilibrium outcome of AI--energy interactions. The same intervention, expanding renewables, can reduce emissions in one regime and accelerate fossil-powered scaling in another. Which regime prevails depends on the economics of AI scaling and market demand; the central policy problem is to ensure that the composition effect dominates the scale effect at the margin.

Accordingly, ``more renewables'' is not a sufficient prescription when scaling is market-led. Partial clean investment may mainly relax a binding compute constraint, accelerating frontier scaling without ensuring that marginal electricity is carbon-free. Avoiding the adaptation trap therefore requires either (1) reducing the cost of the last-mile renewable expansion to $y=k$ (through permitting, interconnection, transmission, storage, and firming) so that condition (b) in \cref{prop:1}(i) is met, or (2) directly raising the private marginal cost of fossil-powered compute (for example, through carbon pricing or clean-energy matching requirements) so that condition (a) in \cref{prop:1}(i) is met. In contrast, when scaling is resource-led, clean capacity is more likely to steer AI toward decarbonization because it directly constrains the developer's capability choice. 
Intermittency, nevertheless, tilts the equilibrium toward carbon and makes the adaptation trap more prevalent, and it elevates operational instruments that firm renewable delivery, including workload flexibility, hourly matching, and storage, into climate-policy levers.

More generally, whether an emphasis on climate adaptation weakens or strengthens mitigation depends on the scaling regime. In broader climate-policy debates, the relationship between adaptation and mitigation is contested, with evidence and theory supporting both substitution and complementarity, depending on institutions and mechanisms \citep{remshard2022}. Our results provide a concrete mechanism that may explain this ambiguity in the AI--energy nexus. In a market-led regime, adaptation motives can rationally support an equilibrium with residual fossil marginal power, weakening mitigation; in a resource-led regime, adaptation motives can instead strengthen the case for clean capacity, making mitigation and adaptation complements. We also note that both mechanisms hinge on one fundamental property of adaptation: its benefits, and therefore the social value of capability, can increase as climate conditions worsen. Our analysis shows that taking into account such interactions can help policymakers understand the implications of adaptation delivered by technologies with different scaling properties.

\section{Calibration and Counterfactuals}
\label{sec:case-study}

We illustrate the model's quantitative implications using four calibrated numerical instances (Instances A--D). Each instance parameterizes the model using observed magnitudes for AI investment, compute intensity, and power-system costs. Instances A--C are informed by conditions in the United States, China, and the European Union over 2023--2025, respectively. Instance D is a stylized counterfactual designed to complement the first three instances and isolate the mechanisms underlying the results.

Instances A and B satisfy the market-led condition $\lambda>1+\alpha$ but differ in market potential (the level of $\theta$), capturing market-led growth with strong versus limited market potential. Instance C satisfies $\lambda<1+\alpha$ with market potential comparable to that of Instance B, capturing resource-led growth with limited market potential. Instance D mirrors Instance C except that $\theta$ takes its Instance A value, producing a resource-led scaling environment with strong market potential. \cref{tab:parameters} summarizes the parameter values for each instance. We briefly describe data sources and calibration below; readers are referred to \cref{ec:calibration} for details. \looseness=-1

\begin{table}[t]
\centering
\small
\renewcommand{\arraystretch}{1.2}
\begin{threeparttable}
\caption{Nominal Parameters Used in the Case Study}
\label{tab:parameters}
\begin{tabular}{@{}l*{4}{c}@{}}
\toprule
Parameter & Instance A & Instance B & Instance C & Instance D \\
\midrule
\\[-8pt]
\multicolumn{5}{@{}l}{\textit{Instance-Specific Parameters}} \\[4pt]
\quad $\theta$ (Billion USD) & 109.08 &   33.65 & 19.42 & 109.08 \\
\quad $\lambda$ & 3.46 & 3.19 & 1.98 & 1.98 \\
\quad $k$ (TWh) & 325.6 &	275.32 &	311.32
 & 311.32 \\
\quad $c_r$ (Billion USD/TWh) & 0.05 & 0.048 & 0.065 & 0.065 \\
\quad $c_f$ (Billion USD/TWh) & 0.088 & 0.151 & 0.146 & 0.146 \\
\quad $V(y)$ (Billion USD) & $0.049y^{1.34}$ & $0.029y^{1.34}$ & $0.036y^{1.34}$ & $0.036y^{1.34}$ \\
\quad $e_f$ (MMTCO$_2$e/TWh) & 0.367 & 0.614 & 0.187 & 0.187 \\
\quad $\xi$ (Billion USD/MMTCO$_2$e)  & 0.04 & 0.0135 & 0.095 & 0.095 \\[10pt]
\multicolumn{5}{@{}l}{\textit{Common Parameters}} \\[4pt]
\quad $\alpha$ & \multicolumn{4}{c}{1.467} \\
\quad $\eta$ & \multicolumn{4}{c}{1.88} \\
\quad $b$ & \multicolumn{4}{c}{0.15} \\
\bottomrule
\end{tabular}
\begin{tablenotes}
\small
\item \textit{Notes:} Instances A and B represent market-led scaling scenarios; Instances C and D represent resource-led scaling scenarios. Instances A and D feature strong market potential; Instances B and C feature limited market potential.
\end{tablenotes}
\end{threeparttable}
\end{table}

\subsection{Calibration}
\label{subsec:constr-case-study}

We start by calibrating the willingness-to-pay function $\theta x^\lambda$. We anchor the market-potential parameter $\theta$ to private investment magnitudes reported in  \citet{maslej2025aiindex}, adjusted to reflect other prominent sources of R\&D in AI. 
To estimate the willingness-to-pay elasticity $\lambda$, we relate frontier-model capability, measured by MMMU (Massive Multi-discipline Multimodal Understanding) scores \citep{yue2024mmmu}, to AI investment across regions over 2023--2025. This yields $\lambda$ values of $3.46$, $3.19$, and $1.98$ for the three regions considered, indicating meaningful heterogeneity in how markets reward marginal capability improvements. We then estimate the economic-spillover parameter $\eta$ by taking the ratio of AI's overall economic value potential \citep{hatzius2023potentially} 
to the profit potential of the AI sector. %

We next calibrate the scaling term $x^{1+\alpha}$ and the energy-demand scale $k$. As a simplification, we assume the major markets share a common scaling exponent $\alpha$ but allow for regional heterogeneity in the level of energy intensity through $k$. We estimate $\alpha$ by regressing regional data-center electricity usage \citep{iea2025} on frontier MMMU capability levels \citep{yue2024mmmu}, obtaining $\alpha=1.467$. We then calculate $k$ in each instance based on the corresponding regional energy-use magnitude. %

Finally, we calibrate the energy and climate primitives. We treat PPA-backed renewable access as the relevant clean procurement margin for AI developers and use data reported by  \citet{iea2025} to set $c_r$ to the cost of renewable PPAs with 80\% hourly matching; we also use the same source's reported grid-price ranges to construct a mean--variance cost for $c_f$, our measure of the (risky) outside option of relying on general grid electricity. Renewable-investment costs take the form $V(y)=g\,y^{1.34}$: the exponent is guided by construction-cost patterns across large Chinese solar projects (Midong, Mengxi Lanhai, Tianwan offshore, and Ordos), whereas $g$ is pinned down by regional average construction costs (including plant, transmission, and land) adjusted for project lifespans. We estimate the climate-damage weight $\xi$ based on the regional carbon prices and set the adaptation parameter $b=0.15$ based on evidence that AI-enabled hazard mitigation can reduce projected average losses by up to 15\% \citep{deloitte2025_ai_infrastructure_report}. Lastly, $e_f$ is set to the regional grid emissions intensity. \looseness=-1

The estimates above should be read with their epistemic status in view: the case study is a calibrated illustration rather than measurement. Parameters describing energy demand, energy costs, emissions intensity, and investment costs are comparatively well measured, but the willingness-to-pay elasticity $\lambda$ rests on three annual observations from a rapidly evolving market per region, and the associated confidence intervals are wide enough that no regional regime classification is statistically firm; alternative willingness-to-pay estimators can reverse the China and EU classifications, although the U.S. classification is robust across estimators (see \cref{subsec:calib-market}). Accordingly, the four instances are best read as scenarios that span the two scaling regimes while remaining anchored to each region's observable magnitudes. Because the design covers both regimes, reclassifying a region relocates it to an instance that is already analyzed; none of the qualitative conclusions rests on any single region's label. \Cref{subsec:calib-market} reports the uncertainty quantification, an estimator comparison, and a measurement agenda.

\subsection{Equilibrium Outcomes}

We solve for equilibrium renewable investment and AI capability choices for the four instances in \cref{subsec:constr-case-study}. \cref{fig:case study} summarizes the outcomes, plotting equilibrium AI energy demand and renewable PPA capacity. We also perform an extensive sensitivity analysis to examine the impact of parametric changes. The key observations are as follows.

\begin{figure}
    \centering
    \begin{subfigure}[b]{0.39\textwidth}
        \centering
        \includegraphics[width=\linewidth]{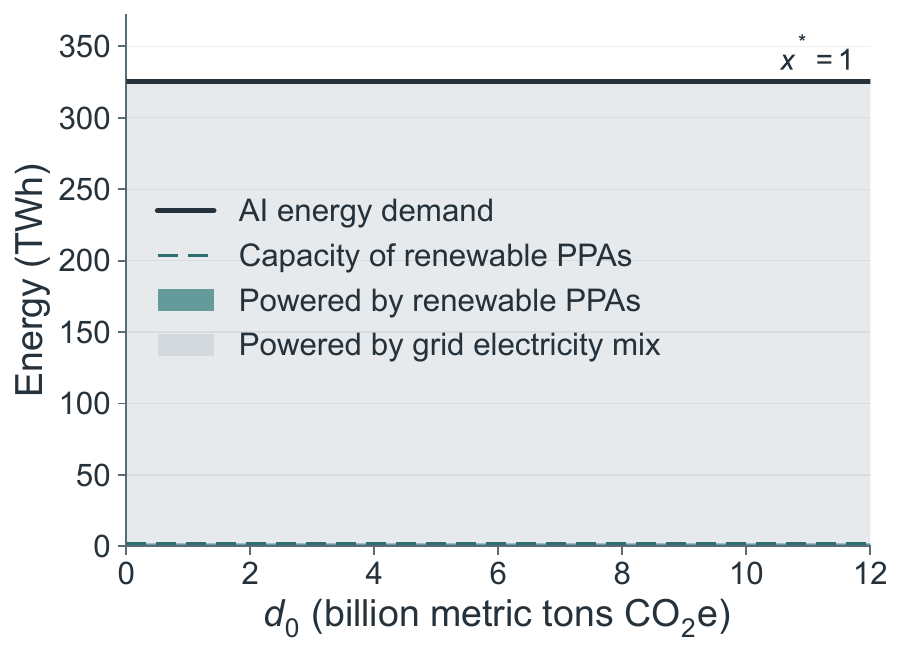}
        \caption{Market-Led Scaling with Strong Market Potential (Instance A)}
    \end{subfigure}
    \hfill
    \begin{subfigure}[b]{0.39\textwidth}
        \centering
        \includegraphics[width=\linewidth]{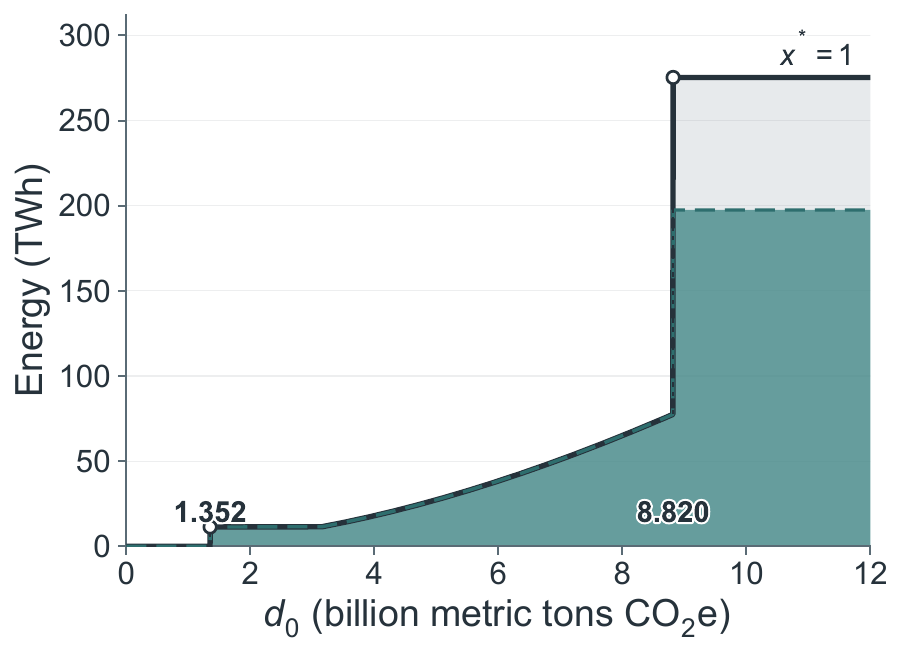}
        \caption{Market-Led Scaling with Limited Market Potential (Instance B)}
    \end{subfigure}

    \vspace{0.5cm}
    \begin{subfigure}[b]{0.39\textwidth}
        \centering
        \includegraphics[width=\linewidth]{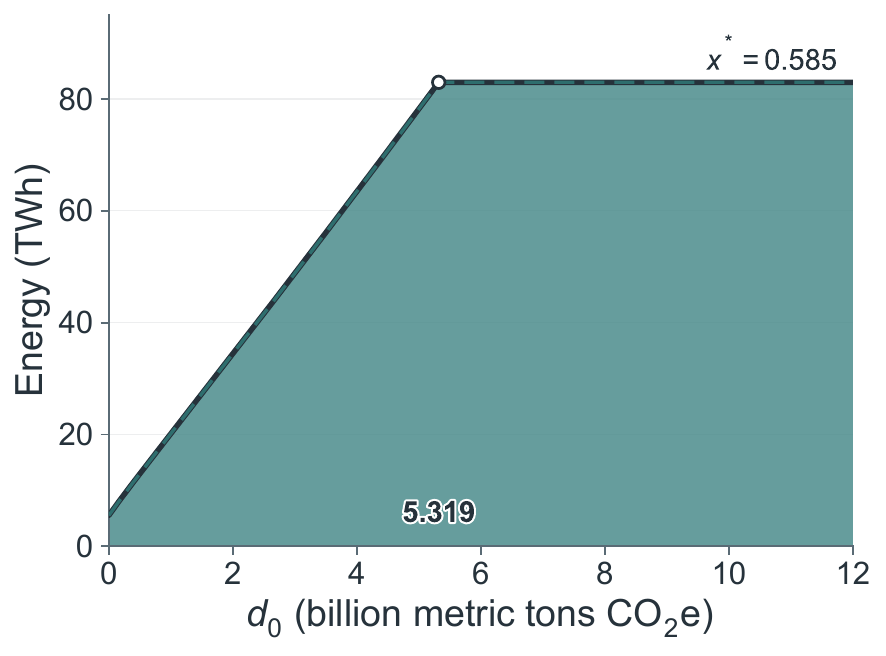}
        \caption{Resource-Led Scaling with Limited Market Potential (Instance C)}
    \end{subfigure}
    \hfill
    \begin{subfigure}[b]{0.39\textwidth}
        \centering
        \includegraphics[width=\linewidth]{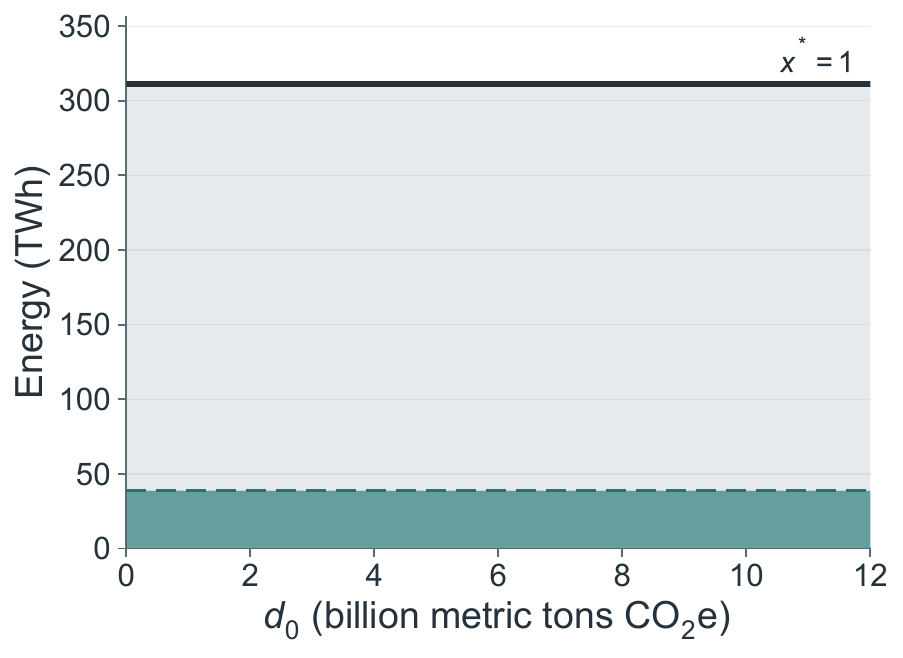}
        \caption{Resource-Led Scaling with Strong Market Potential (Instance D)}
    \end{subfigure}
    \caption{AI Energy Demand and Renewable Energy Investment Outcomes in the Case Study. Panels plot AI energy demand and renewable PPA capacity against the baseline emissions stock $d_0$ (billion metric tons CO$_2$e; vertical scales differ). $x^*$ values are annotated at the right of each panel and apply to the region beyond the last marked threshold; circled markers with adjacent labels give regime-transition thresholds. Renewable PPAs meet $0.6\%$ of demand in Instance A and $12.47\%$ in Instance D.}
    \label{fig:case study}
\end{figure}

\subsubsection*{Instance A}
First, market-led scaling with strong market potential places the equilibrium in a decoupling zone, shown in Panel (a) of \cref{fig:case study}, in which renewable investment remains far below energy demand. In this case, the adaptation trap becomes intrinsic to AI development: it arises even without external emissions (i.e., when $d_0=0$). This is because, with a high $\theta$, the developer is sufficiently motivated to adopt $x^*=1$ even when most of the energy demand is met by the grid. This leads to substantial emissions from AI, and the adaptation benefits of frontier capability in addressing those emissions sustain the carbon-intensive frontier equilibrium. 

Sensitivity analyses indicate that restoring coupling requires substantial changes in the operating environment so that renewable capacity materially affects the developer's choice of AI scale. One route is to increase the cost of fossil energy, $c_f$, by 283.6\% (under which the coupling zone in \cref{lem:1} exists). Alternatively, achieving a clean share of at least 50\% in the decoupling zone, where the clean share is the fraction of AI energy demand met by renewable PPAs, requires reducing the clean-capacity cost scale $g$, by at least 77.4\%. Improvements in energy efficiency, represented by a reduction in $k$, can lower the required reduction in $g$. %

\subsubsection*{Instance B}
Second, market-led scaling with limited market potential leads to coupling when $d_0$ is low. Renewable capacity can therefore expand before the equilibrium transitions to the carbon-intensive frontier, leaving a comparatively large renewable base even after decoupling occurs. This explains why frontier AI in Panel (b) of \cref{fig:case study} is supported by a clean share of approximately 71.7\%, close to the 80\% policy goal discussed for China \citep{tnw2026china}. Nevertheless, the transition into the adaptation trap occurs at a relatively low level of external climate pressure. The threshold $d_0=8.820$ billion metric tons of CO$_2$e corresponds to an increase of approximately 1.13 ppm in atmospheric carbon concentration above the climate benchmark \citep{noaa_gml_ppm}. At recent rates of atmospheric carbon accumulation, this increase is equivalent to roughly six months of additional concentration growth \citep{lindsey2025co2}. 
As another illustration, the intensity of extreme events has risen markedly since 2020 \citep{guardian2025weather}. Benchmarking the 2025 annual mean atmospheric CO$_2$ (425.6 ppm) against the 2020 annual mean (412.4 ppm) implies a $d_0$ of approximately 103 billion metric tons of CO$_2$e, far exceeding the threshold in Panel (b) \citep{noaa2026trends}.
(The panels' common horizontal range focuses on the transition region; at benchmark values of this magnitude, every instance sits far beyond its threshold, where outcomes no longer vary with $d_0$.) %

Sensitivity analysis shows that a further increase in market willingness-to-pay in Instance B likely accelerates the transition to frontier AI and the trap mechanism, and sharply reduces the clean share once in that region: increasing $\theta$ by only about 5\% reduces the clean share to 45.4\%. One effective countermeasure is increasing the economic-spillover parameter $\eta$ (e.g., general productivity improvements from AI deployment via better human-AI collaboration): a 30\% increase in the spillover component $\eta-1$ restores the clean share to 66\% even when $\theta$ is much larger, and a 70\% increase in $\eta-1$ helps achieve frontier AI powered by 100\% renewable energy. Increasing the fossil-energy price can produce similar effects. 

\subsubsection*{Instances C and D}
Third, under resource-led scaling, AI expansion remains closely tied to renewable availability because developers are more sensitive to energy costs. With limited market potential, the equilibrium remains coupled even at $d_0=0$, as shown in Panel (c) of \cref{fig:case study}, yielding an adaptation pathway that operates across all levels of climate stress. This coupling, however, limits AI development: under the baseline calibration, capability plateaus at $x^*=0.585$. %

A moderate increase in market willingness-to-pay can relax this constraint without immediately breaking the coupling between AI and renewable capacity. For example, when $\theta$ in Instance C is increased to the level used in Instance B, the equilibrium remains coupled and reaches $x^*=1$ once $d_0$ is sufficiently high (which can be interpreted as a sufficiently stringent climate policy that adopts a low carbon-concentration benchmark when evaluating $d_0$). Further increases in $\theta$, however, eventually create a decoupling zone. At intermediate levels of market potential, the pathway mechanism can restore coupling, but only when $d_0$ is sufficiently high. When $\theta$ reaches the level used in Instance A, the pathway mechanism disappears altogether, and the equilibrium remains in the decoupling zone with a low clean share, as shown in Panel (d). Thus, stimulating AI market demand can help overcome the capability constraint under resource-led scaling, but excessive expansion can break the link between AI growth and renewable development. Improving energy efficiency offers a less destabilizing alternative: in the baseline Instance C, reducing $k$ by approximately 23\% is sufficient to achieve $x^*=1$ while preserving coupling.

Finally, Instance C is a moderately resource-led scaling scenario in which $\frac{1+\alpha}{\lambda}<\frac{c_f}{c_r}$ holds and is therefore particularly vulnerable to the impact of renewable intermittency. That is, coupling could lead to higher emissions, and the adaptation pathway may reverse into a trap mechanism; in Instance C, this reversal arises when the availability probability falls below roughly $0.98$ (which is the $\bar{p}$ threshold in \cref{lem:intermittency}; see \cref{subsec:utilization}), so firming is essential, rather than merely helpful, for preserving the pathway.

\subsection{Discussion}
\label{subsec:discussion}

Taken together, the case study shows that the adaptation-trap and adaptation-pathway mechanisms are practically relevant but could manifest in different ways depending on the scaling environment. Notably, high market potential effectively entrenches the adaptation trap and suppresses the adaptation pathway, thereby leading, under both market-led and resource-led scaling, to decoupling in which the energy demand from AI development substantially outpaces renewable investment. A moderate market potential, on the other hand, can lead to different outcomes between the two scaling regimes: in the ranges represented by Instances B and C, market-led scaling tends to sustain frontier capability but fall
short of net zero under the adaptation trap, whereas resource-led scaling tends to deliver cleaner energy at the cost of
frontier scale under the adaptation pathway. This illustrates a practical tension between AI development and the
clean-energy transition.

Although the ``power-couple'' outcome, that is, frontier capability supplied fully by renewable electricity ($x^*(y^*)=1$ with $y^*=k$), does not arise in the baseline calibrations, the sensitivity analyses identify several mechanisms that can move the equilibrium in that direction. Importantly, the binding constraint differs across the two scaling regimes. Under market-led scaling, as in Instances A and B, renewable development is the primary bottleneck: AI's capability expands readily, but clean capacity lacks sufficient influence over scaling. Effective interventions must therefore strengthen the link between additional AI scale and additional clean power, for example, by lowering renewable-investment costs, increasing the price of fossil energy, strengthening climate incentives, improving energy efficiency, or increasing the general economic value generated by AI development. By contrast, under resource-led scaling, as in Instance C, AI capability is the primary bottleneck. The relevant interventions should therefore strengthen incentives for frontier capability without undermining coupling, such as moderate increases in market potential and AI energy-efficiency improvements. \Cref{ec:addresults} provides an analytical characterization of this contrast: in all scenarios, the clean-capacity cost scale $g$, must lie below a threshold for a power-couple outcome to materialize. That threshold, and with it the effectiveness of different interventions, depends on whether frontier capability is achieved in the decoupling zone, where renewable development may lag behind (Instances A and B), or in the coupling zone, where AI development may be constrained (Instance C).

Despite these differences, the results point to a common conclusion: net-zero frontier AI is unlikely to result from any single intervention. Achieving the AI--renewable power couple requires a coordinated portfolio: cheaper and easier-to-integrate clean capacity, policy that sufficiently internalizes climate damages and broader externalities of fossil energy, and management of AI market expansion and AI-driven electricity demand in a manner consistent with the prevailing scaling regime. Even then, feasibility is limited in strongly market-led environments. In particular, rapid growth in investment can weaken coupling, shifting the equilibrium away from net-zero frontier AI.

\section{Extensions}
\label{sec:extensions}
We relax the baseline assumptions along four dimensions. Two concern the AI side of the market and are developed here: capability may lower the cost of integrating renewables (\cref{subsec:reductions-renewable-costs}), and development may be competitive rather than monopolistic (\cref{subsec:competition}). Two concern how clean capacity is financed and physically delivered, and are developed in the electronic companion: \cref{subsec:decentralized} replaces the policymaker's direct capacity choice with a financial incentive to which a private utility responds, and \cref{subsec:utilization} allows renewable output to be intermittent rather than dependable and analyzes the firming that counters intermittency.\footnote{Fossil-side policy requires no separate extension because it operates through $c_f$: fossil subsidies lower $c_f$ and carbon pricing raises it (a full welfare accounting would also net out carbon-tax revenue as a transfer, which our welfare function does not model), shifting the regime thresholds (for example, the condition in \cref{prop:1}(i) and the capability target $f(c_f)$ in \cref{lem:2}) without altering the two-regime structure. The counterfactual discussion in \cref{subsec:discussion} considers this lever. }  All four clarify when the equilibrium tilts toward the adaptation trap versus the adaptation pathway. Hence, we focus on the market-led scaling scenario in this section.

\subsection{AI-Enabled Reductions in Clean-Capacity Costs}
\label{subsec:reductions-renewable-costs}

The main model assumes renewable-investment costs $V(y)$ are independent of AI capability. In practice, higher capability can lower the cost of deploying and integrating renewables, in particular by improving forecasting, grid control, and storage operations, thereby easing integration frictions \citep{Choi2024}. We capture this by allowing
\[
V=V(x,y), \qquad \frac{\partial V(x,y)}{\partial x}<0.
\]
One such specification is
\[
V(x,y)=\phi(x)\,V(y),
\]
where $\phi'(x)<0$ and $\phi(x)<1$ for all $x\in(0,1]$, and $V(y)$ is the baseline cost function used in the main model.

This extension introduces a feedback loop: higher $x$ lowers the marginal cost of $y$, strengthening incentives to invest in renewables, relaxing the energy constraint, and potentially supporting higher $x$. As a result, full decarbonization is easier to sustain; in particular, the renewable-cost threshold for net-zero frontier AI is relaxed by a factor of $1/\phi(1)$ (see \cref{subsec:td-costreduction}).

\begin{figure}
    \centering
\begin{subfigure}[b]{0.32\textwidth}
        \centering
        \includegraphics[width=\linewidth]{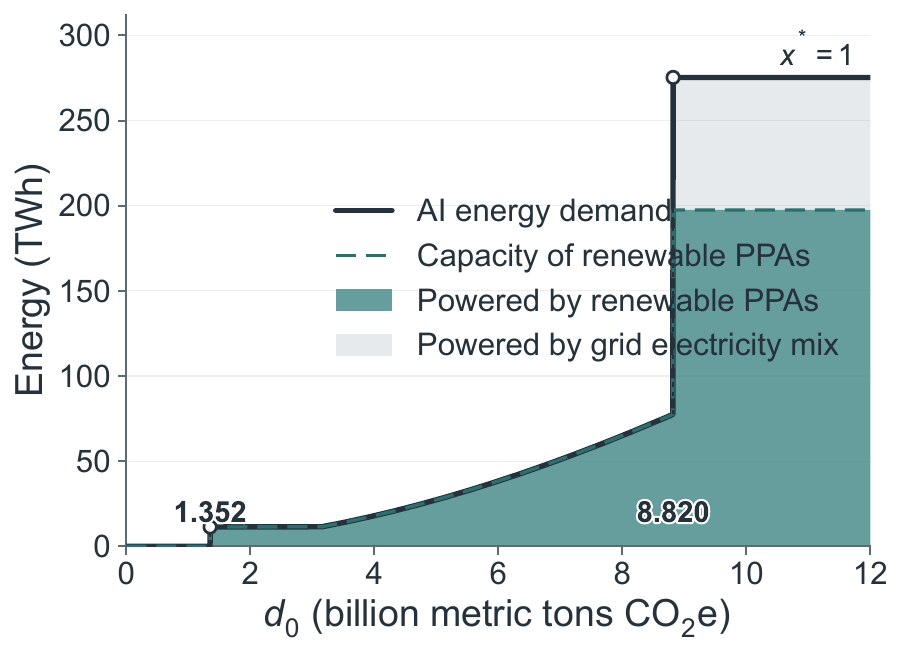}
        \caption{$\phi(x)=1$ (the main model)}
    \end{subfigure}
\begin{subfigure}[b]{0.32\textwidth}
        \centering
        \includegraphics[width=\linewidth]{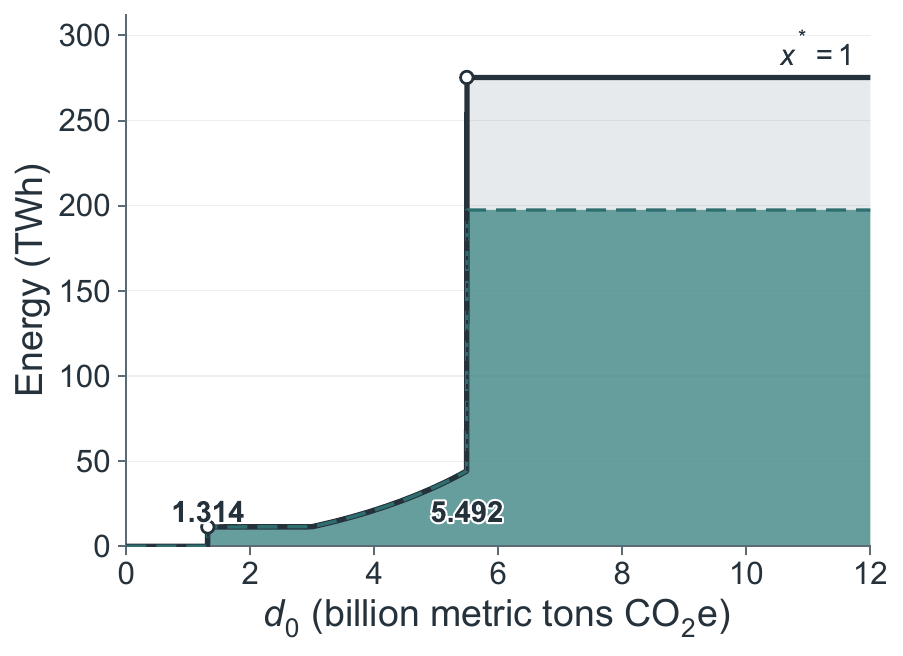}
        \caption{$\phi(x)=1-0.1x$ }
    \end{subfigure}
    \begin{subfigure}[b]{0.32\textwidth}
        \centering
        \includegraphics[width=\linewidth]{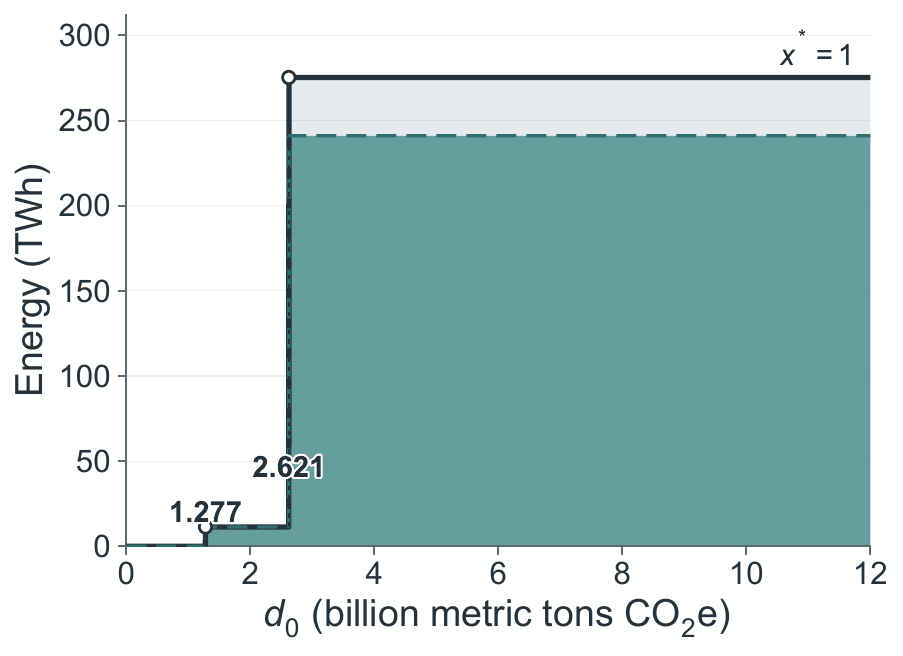}
        \caption{$\phi(x)=1-0.2x$ }
    \end{subfigure}
    \caption{Adaptation Trap under AI-Enabled Renewable Investment Cost Reduction. Illustration based on the parametric setting of Instance B in the case study.}
    \label{fig:extension1}
\end{figure}

However, when AI-enabled cost reductions are insufficient to deliver full decarbonization, they can instead intensify the adaptation trap by expanding the parameter region in which the mechanism operates. \Cref{fig:extension1} illustrates this pattern using Instance~B from the case study. With an AI-enabled cost-reduction potential of roughly \(10\%\) (Panel~(b)), an adaptation trap emerges once \(d_{0}\) reaches \(5.492\). This cutoff is approximately 38\% lower than the corresponding threshold in the absence of AI-enabled cost reduction, \(d_{0}=8.820\), shown in Panel~(a), which reproduces \Cref{fig:case study}(b) for comparison. Increasing the cost-reduction potential to \(20\%\) lowers the threshold further, as shown in Panel~(c). The intuition is straightforward: AI-enabled cost reductions make frontier scaling more attractive in equilibrium, thereby accelerating decoupling and, under market-led scaling, bringing forward the conditions under which the adaptation-trap mechanism arises.

Another key observation from comparing \cref{fig:extension1}(a) and (b) is that AI-enabled cost reductions may not translate into increased renewable-capacity investment once the equilibrium enters the adaptation trap. This outcome arises from the loss of the scale effect as developers' capability choices decouple from renewable availability. Consequently, only sufficiently large cost reductions can compensate for the loss of this lever and restore incentives for additional capacity investment. For Instance~B, illustrated in \cref{fig:extension1}, this cost-reduction requirement is at least 14.41\%. If this requirement is met, AI-enabled cost reductions entail a trade-off in which the adaptation trap is more likely to occur but less fossil-dependent, as \cref{fig:extension1}(c) shows. Otherwise, such cost reductions tend to entrench carbon reliance by accelerating the same carbon-intensive trap, as is the case in \cref{fig:extension1}(b).

In practice, AI-enabled improvements in renewable development and integration are often cast as both an accelerant of the energy transition and a way to shrink AI's own footprint. Our analysis shows that these gains depend on equilibrium responses: in some environments, efficiency improvements can deepen reliance on carbon-intensive supply rather than displace it. A decarbonizing AI--energy relationship requires more than technological complementarity: investment-side coupling must be sustained. In particular, policy may need to preserve the scale channel so that incentives for renewable-capacity expansion keep pace with AI-driven electricity demand.

These effects are one instance of a broader pattern that spans the renewable-promoting levers available to a policymaker. We examine two further levers in \cref{subsec:decentralized,subsec:utilization}: decentralized investment, in which the policymaker sets a financial incentive $s$ and a utility then chooses capacity, and a firmed utilization factor $\zeta$ governing how much installed capacity becomes dependable renewable supply. %
Under market-led scaling,  $s$  and $\zeta$ reduce the carbon intensity of frontier AI yet can hasten the adaptation trap, lowering $\bar{d}_m$ (for $\zeta$, only when it is not large enough to forestall the trap entirely). Under resource-led scaling, both uniformly reinforce the adaptation pathway, lowering $\bar{d}_c$. %
The broader lesson is that making renewables cheaper, more usable, or better subsidized sharpens the incentive to scale capability, so the emissions consequences turn on the scaling regime rather than on the lever itself.

\subsection{Competition Among AI Developers}
\label{subsec:competition}

The main model analyzes a single developer representing the industry. We extend the model to explore the implications of competition between developers. 
Specifically, we consider two AI developers that compete by choosing capabilities $(x_1,x_2)$ first and then prices $(p_1,p_2)$. We refer readers to \cref{subsec:td-competition} for model details. We note that incorporating the price decision does not affect the main-model results in the monopoly case, so the difference observed in this extension can be attributed to competition. 

Competition can reduce the private return to incremental capability through two channels. First, market sharing lowers the effective market potential, $\theta$, available to each developer. Second, competition can reduce the effective elasticity of equilibrium revenue with respect to capability, so that marginal revenue grows more slowly as both developers scale. Together, these effects lower the return to further capability expansion and increase developers' sensitivity to energy costs. Equilibrium capability therefore becomes more responsive to renewable capacity, increasing the likelihood that renewable investment moderates AI expansion and strengthens coupling. Competition can thus reduce the risk of the adaptation trap.%

A caveat, however, is that whether and to what extent this potential can be realized depends on market conditions. \Cref{fig:competition} illustrates a key finding based on extensive numerical experiments. When market potential ($\theta$) is high, the capability decision of either developer remains primarily driven by the market, and the coupling with renewable capacity that competition enables is not strong enough to structurally shift the equilibrium. In this case, the equilibrium exhibits a ``dual adaptation trap'' under which both developers are subject to this mechanism. The dark-blue region in \cref{fig:competition} illustrates when such an equilibrium emerges; Instance A in \cref{sec:case-study} becomes a case in point when the market potential $\theta$ rises above  \$170.5 billion (assuming $d_0=10$), a plausible scenario considering that AI investment in the U.S. reached \$285.88 billion in 2025 \citep{Sajadieh2026aiindex}. At the other extreme, reduced market investment $\theta$ may leave revenue potential insufficient to sustain competition, yielding a monopoly equilibrium (see the white region in \cref{fig:competition}). Instance B corresponds to this case. %

\begin{figure}
    \centering
    \includegraphics[width=0.62\linewidth]{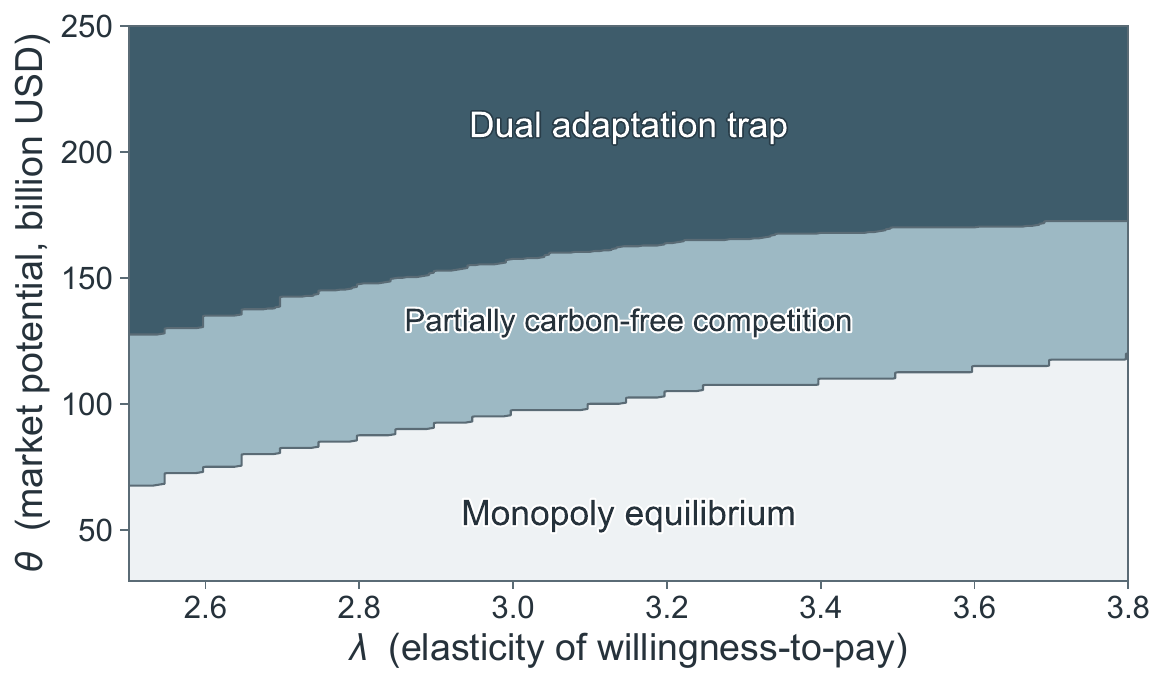}
    \caption{Equilibrium under the Competitive Setting. Illustration based on the parametric setting of Instance A in the case study (except $(\theta,\lambda)$), with $d_0=10$.}
    \label{fig:competition}
\end{figure}

Nevertheless, when $\theta$ is intermediate (see the light-blue region in \cref{fig:competition}), competition can induce partially carbon-free AI development: in equilibrium, at least one developer's energy demand is fully met by renewable resources despite market-led scaling. When $\lambda$ drops below $1+\alpha$, scaling becomes resource-led and the adaptation pathway follows. This extension highlights that an analogous mechanism may already be at work prior to this transition, operating through the competition lever. The effect also persists as $\lambda$ rises over a meaningful range, suggesting that competition, supported by a market with an intermediate yet substantial level of investment, can provide an effective channel for countering the adaptation trap even under strong market-led scaling.

A policy implication follows. Antitrust and market design in AI can shape the carbon intensity of the capability race by altering developers' incentives to scale. When rents are highly concentrated, a dominant developer has stronger incentives to expand compute despite reliance on carbon-intensive power. Greater competition weakens these returns to scale and increases sensitivity to energy costs, making renewable capacity a more effective instrument for shifting AI development toward a carbon-free pathway. \looseness=-1

\section{Conclusion}
\label{sec:conclusions}

AI and renewable energy are increasingly cast as a ``power couple,'' on the premise that surging AI demand will accelerate clean-energy investment and hasten decarbonization \citep{wef2025powercouple,rmi2025powercouple,iea2023powercouple}. This paper examines when that premise holds and when it fails. Rather than taking renewable supply or AI demand as given, we treat capability choices and renewable-capacity expansion as jointly determined and show that equilibrium outcomes depend on two primitives: how compute maps into capability and how capability maps into private value.

A key insight from our analysis is that expanding renewable capacity need not decarbonize AI growth. When the market value of additional capability rises at least as fast as its energy requirement, developers have strong incentives to push the frontier even when the marginal megawatt-hour is fossil-based. When frontier capability is relatively insensitive to energy costs, clean-capacity expansion can cease to affect AI scale before fossil electricity is fully displaced. Past that point, additional renewable investment mainly replaces fossil generation, and that gain may not cover the investment cost; the equilibrium can stop short of full decarbonization even as AI remains at the frontier. %
This regime supports a carbon-intensive equilibrium and an ``adaptation trap'': as climate damages rise, the value of AI-enabled adaptation increases, which strengthens incentives to enable frontier scaling on fossil energy, thereby raising emissions. By contrast, when the market value of additional capability rises more slowly than its energy requirement, energy costs become binding, and capability responds more directly to renewable capacity; clean investment then both enables capability and decarbonizes marginal compute, yielding an ``adaptation pathway'' in which adaptation advances through AI scaled on carbon-free power. \looseness=-1

These regimes also clarify when ``rebound'' is benign and when it is harmful. Expanding renewable supply and improving energy efficiency both lower the effective cost of compute and can accelerate scaling. Under market-led scaling, this acceleration loosens the link between compute and clean capacity, slowing decarbonization and increasing the risk of carbon lock-in. Under resource-led scaling, the same cost reductions tighten that link, making renewable investment and AI growth mutually reinforcing, provided market potential remains in the coupled range. A carbon-intensive equilibrium, therefore, can arise from incentives and scaling economics alone, not only from slow infrastructure or political inertia. A clean transition, on the other hand, is not guaranteed by ongoing renewable expansion, the cost advantage of clean power, or AI-driven efficiency improvements; it requires policy and market design that preserve the link between AI scaling and clean-capacity expansion.

Four policy lessons follow. First, policies that expand clean capacity without addressing marginal emissions need not lower AI's emissions intensity; policies that raise the private cost of fossil-powered compute can be instrumental, whether through carbon pricing, carbon-free procurement requirements, or constraints tied to local grid emissions. Second, to strengthen the coupling between AI growth and clean-energy deployment, policy must govern the conditions under which AI scales: the location and timing of large loads, renewable interconnection, and rules for carbon-free matching. Third, when market forces are exceptionally strong, energy policy may have limited leverage; pursuing additional emissions reductions in other sectors to offset AI-related emissions may be more effective, though evaluating cross-sector instruments lies beyond our single-sector model. Finally, the welfare value of AI's adaptation potential is state-dependent: as climate damages rise, the policymaker's optimal response can switch sharply between the trap and the pathway, implying that climate policy should treat adaptation and mitigation as a joint design problem.

Whether AI and renewable energy reinforce each other depends on the economics of AI scaling and on how compute growth connects to the energy system. The energy system's equilibrium response settles AI's environmental consequences, and institutions shape that response. Procurement, interconnection constraints, and the siting and timing of data-center expansion determine whether scaling is met by clean additions or by greater reliance on fossil generation. A clean scaling path requires policies and market design that sustain investment-side coupling so that renewable deployment keeps pace with compute growth. This, in turn, motivates a research agenda at the intersection of operations, markets, and political economy: measuring how workload flexibility affects marginal emissions; designing grid-aware siting and scheduling that internalize spatial and temporal system conditions; studying how institutions such as PPAs, interconnection rules, capacity markets, and hourly matching requirements shape clean investment and the incentives to scale frontier AI; and quantifying the substitution between compute flexibility and energy flexibility, including the latency tolerance that governs how far AI workloads can follow renewable availability. %

\ACKNOWLEDGMENT{The authors are grateful to Vishal Agrawal, Saurabh Amin, Atalay Atasu, Hamsa Bastani, John R.\ Birge, Charles J.\ Corbett, Jan C.\ Fransoo, Huseyin Gurkan, Serguei Netessine, Masha Shunko, Owen Q.\ Wu, and \c{S}afak Y\"ucel for their thoughtful comments on and reactions to a preliminary draft, and to the anonymous reviewers of the 2026 MSOM Sustainable Operations SIG for their constructive feedback. They are especially indebted to Serguei Netessine for his excellent discussion of the paper at the Sustainable Operations SIG of the 2026 MSOM Conference held at Harvard Business School, and they thank the session participants for their helpful questions. They also thank participants at the Supply Chain Management Conference on Sustainability and Social Responsibility in Supply Chains, held at the Gies College of Business, University of Illinois Urbana-Champaign, and at the Southern California OR/OM Day at USC Marshall, where earlier versions of this paper were presented.
}

\let\oldbibliography\thebibliography
 \renewcommand{\thebibliography}[1]{%
 	\oldbibliography{#1}%
 	\baselineskip16pt %
 	\setlength{\itemsep}{0pt}%
 }

\bibliographystyle{informs2014TD}

\bibliography{ai-climate}

\ECSwitch
\makeatletter
\def\theHsection{EC.\arabic{section}}
\def\theHsubsection{EC.\arabic{section}.\arabic{subsection}}
\def\theHsubsubsection{EC.\arabic{section}.\arabic{subsection}.\arabic{subsubsection}}
\def\theHproposition{EC.\arabic{proposition}}
\def\theHcorollary{EC.\arabic{corollary}}
\def\theHlemma{EC.\arabic{lemma}}
\def\theHtheorem{EC.\arabic{theorem}}
\def\theHdefinition{EC.\arabic{definition}}
\def\theHequation{EC.\arabic{equation}}
\def\theHtable{EC.\arabic{table}}
\def\theHfigure{EC.\arabic{figure}}
\makeatother

\begin{center}
{\Large \bfseries Electronic Companion to}\\[0.25em]
{\large \textit{``\PaperTitle''}}
\end{center}

\small
\setcounter{footnote}{0}

This electronic companion is organized in two parts. \Cref{ec:addresults,ec:addextensions,ec:techdetails}
develop supplementary economics: conditions under which frontier AI is powered entirely by renewables,
two extensions that relax the baseline assumptions on who invests and on what installed capacity delivers,
and the formal details behind the two extensions analyzed in the main text.
\Cref{ec:proofs,ec:calibration} provide technical verification: proofs of the four results stated in the main text,
followed by the calibration sources and computations behind \cref{tab:parameters} and the case study.
Proofs of the companion's own results accompany their statements.

\section{Additional Analytical Results: Conditions for Carbon-Free Frontier AI}
\label{ec:addresults}

In the counterfactual discussion in the case study (\cref{subsec:discussion}), we examine conditions under which an equilibrium arises in which the optimal AI capability choice is $x^*(y^*)=1$ and the optimal renewable capacity investment is $y^*=k$. Analytically, it can be shown that such conditions may involve (i) sufficiently low renewable investment cost, (ii) energy demand that is not excessively high, and (iii) sufficiently adverse climate conditions. For simplicity, we adopt $V(y)=g\cdot y^\mu$ where $\mu>1$ in this analysis.

\begin{proposition}
\label{prop:3}
    In the market-led scaling scenario, there exists an energy demand threshold $\bar{k}$, a climate condition threshold $\hat{d}$, and two marginal renewable investment cost thresholds $g_1,g_2$ such that an equilibrium in which $x^*(y^*)=1$ and $y^*=k$ is achieved if (i) $k\leq \bar{k}$, $d_0\geq \hat{d}$, and $g<g_1$ when $\frac{\theta}{k}\leq \max\{\frac{c_f(1+\alpha)}{\lambda},c_r\}$, or (ii) $g\leq g_2$ when $\frac{\theta}{k}> \max\{\frac{c_f(1+\alpha)}{\lambda},c_r\}$. %
    \label{prop:EC1}
\end{proposition}

\noindent 
\emph{Proof of \cref{prop:EC1}.} First, when $\frac{\theta}{k}> \max\{\frac{c_f(1+\alpha)}{\lambda},c_r\}$, according to the fourth paragraph in the proof of \cref{prop:1}, $x^*=1$ materializes only in the ``decoupling zone'' characterized by the third scenario in \cref{lem:1}. Hence, achieving $y^*=k$ first requires that $y=k$ maximizes welfare along the $x^*(y)=1$ branch. This happens when  $V'(k)\leq \eta(c_f-c_r)+(1-b)e_f\xi$, which is equivalent to $g\leq g^0_2\doteq \frac{\eta(c_f-c_r)+(1-b)e_f\xi}{\mu k^{\mu-1}}$. %
Second, we need to ensure that the welfare under $y=k$, $W(k)=\eta(\theta-c_rk)-\xi(1-b)d_0-gk^{\mu}$, is no smaller than that in the other two scenarios in \cref{lem:1}: on the branch where $x^*(y)=0$, welfare equals $-\xi d_0-gy^{\mu}$ and is maximized at $y=0$, so $W(0)=-\xi d_0$; $W(k)\geq W(0)$ if $g\leq g_3^0\doteq \frac{\eta(\theta-c_rk)+\xi b d_0}{k^{\mu}}$. On the branch where $x^*(y)=(y/k)^{1/(1+\alpha)}$,  %
the requirement $W(k)\geq W(y)$ is equivalent to
\begin{equation*}
g\left(k^{\mu}-y^{\mu}\right)\leq \eta\left[\theta\left(1-\left(\tfrac{y}{k}\right)^{\frac{\lambda}{1+\alpha}}\right)-c_r(k-y)\right]+\xi b d_0\left[1-\left(\tfrac{y}{k}\right)^{\frac{1}{1+\alpha}}\right].
\end{equation*}
This inequality holds at every such $y$ on this branch if $g\leq g^0_4$, where 
\begin{equation*}
g_4^0\doteq \inf_{y_1\leq y<y_2}\frac{\eta\left[\theta\left(1-\left(\tfrac{y}{k}\right)^{\frac{\lambda}{1+\alpha}}\right)-c_r(k-y)\right]+\xi b d_0\left[1-\left(\tfrac{y}{k}\right)^{\frac{1}{1+\alpha}}\right]}{k^{\mu}-y^{\mu}}
\end{equation*}
Imposing $g\leq g_2\doteq \min\{g^0_2,g^0_3,g^0_4\}$ therefore proves \cref{prop:3}(ii). %

Second, when $\frac{\theta}{k}\leq \max\{\frac{c_f(1+\alpha)}{\lambda},c_r\}$, according to the second paragraph in the proof of \cref{prop:1}, two possibilities arise. 
\begin{itemize}
    \item $\frac{c_r}{c_f}<\frac{1+\alpha}{\lambda}$ and $\frac{\theta}{k}\leq \frac{1+\alpha}{\lambda }c_f$: in this case, $y_2=k$. Hence, an equilibrium in which $x^*(y^*)=1$ and $y^*=k$ can be achieved only in the ``coupling zone'' in which $x^*(y)=(\frac{y}{k})^{1/(1+\alpha)}$. Plugging this solution into the welfare function leads to a function that we denote by $W_3(y)$. It can be shown that this equilibrium is indeed achieved if (a) $d_0\geq \frac{\eta  \theta  \lambda  (\lambda- (1+\alpha))}{\alpha  b \xi }$, and (b) $g\leq \min\{\frac{b d_0 \xi -(1+\alpha) c_r \eta  k+\eta  \theta  \lambda }{(\alpha +1) k\cdot \mu k^{\mu-1}},\frac{b d_0 \xi -c_r \eta  k+\eta  \theta}{k^{\mu}}\}$. Condition (a) ensures that $W_3(y)$ is concave, so the optimal solution is achieved either at a stationary point or at the endpoint. Thus $y^*=k$ when $W_3'(k)\geq 0$ and $W_3(k)\geq W(0)$ (i.e., renewable investment yields higher welfare than not investing). The two terms that define the $g$ threshold in condition (b) ensure that these two inequalities hold.

    \item $\frac{c_r}{c_f}\geq \frac{1+\alpha}{\lambda}$ and $\frac{\theta}{k}\leq c_r$: in this case, $x^*(y)=0$ for any $y\in[0,k]$. Hence, an equilibrium in which $x^*(y^*)=1$ and $y^*=k$ cannot be achieved. We impose $k<\frac{\theta}{c_r}$ to eliminate this case. 
\end{itemize}

Combining the conditions on $(k,d_0,g)$ proves \cref{prop:3}(i). \hfill \emph{Q.E.D.}

\smallskip

Although it is intuitive that ensuring net-zero frontier AI requires a sufficiently low marginal cost of investment in renewables, Proposition \ref{prop:3} indicates that the cost threshold depends on the ratio $\frac{\theta}{k}$, which measures the average market investment per unit of energy demand. If this ratio exceeds $\frac{c_f(1+\alpha)}{\lambda}$ (which is equivalent to $\theta\lambda\geq c_fk(1+\alpha)$, i.e., the marginal market value exceeds the marginal cost at $x=1$), it can be shown that maximum AI capability is achieved in the decoupling zone, independent of whether $y=k$. In this case, the reduction of the marginal cost of renewable investment is driven only by the composition effect, that is, it should be lower than the marginal benefit from displacing fossil fuel. Indeed, the proof shows that  $g$ should be no higher than the cost threshold identified in condition (b) in \cref{prop:1}(i). 

Otherwise, maximum AI capability may be achieved in the coupling zone, which requires full decarbonization to justify its cost. In this case, the marginal benefit of renewable investment is determined by the scale effect, which includes the developer's scale margin, the associated economic spillover, and the enhanced adaptation benefits, minus the renewable energy generation cost. These components determine the $(d_0,g_1)$ thresholds. A caveat is that in this case, the marginal benefit of renewable investment could be negative when frontier AI entails too high an energy cost even when powered by the cheaper renewable sources and therefore cannot be realized. This explains the upper bound on the energy-demand scale $k$ for this case. Another useful observation is that because the scale effect draws on AI's adaptation benefits, the $g_1$ bound can be much more relaxed than $g_2$ when $d_0$ is high, which suggests an avenue to circumvent the limited renewable investment cost reduction potential, which is constrained by technology.

We show that \cref{prop:3} applies structurally to  resource-led scaling scenarios. This
is because the impact of the comparison between the elasticity parameters $\lambda$ and $1+\alpha$ is attenuated at frontier
capability $x=1$. The next proposition formalizes the finding. 

\begin{proposition}
    In the resource-led scaling scenario, \cref{prop:3} continues to hold if (a) the threshold on $\frac{\theta}{k}$ that separates cases (i) and (ii) is changed to $\frac{c_f(1+\alpha)}{\lambda}$, (b) the thresholds on $(k,d_0)$ are removed,  and (c) the $g_1$ threshold is replaced by a different one, $g_3$, and the $g_2$ threshold is replaced by $g_2^0$ defined in the proof of \cref{prop:EC1}. %
    \label{pro:EC.2}
\end{proposition}

\noindent \emph{Proof of \cref{pro:EC.2}.} Based on \cref{lem:2}, it can be shown that when $\frac{\theta}{k}>\frac{c_f(1+\alpha)}{\lambda}$, $f(c_f)>1$, thus only the first case (the ``decoupling zone'') in \cref{lem:2} materializes, in which $x^*(y)=1$ for any $y\in[0,k]$. Applying the same proof for \cref{prop:3}(ii) proves the result for the resource-led scaling scenario. 

When $\frac{\theta}{k}\leq \frac{c_f(1+\alpha)}{\lambda}$, $f(c_f)\leq 1$, thus the welfare function has at least two pieces. Nevertheless, it can be shown that the equilibrium $x^*=1$, $y^*=k$ can only be achieved in the second case, that is, the ``coupling zone'' of \cref{lem:2}. Plugging the $x^*(y)=(\frac{y}{k})^{1/(1+\alpha)}$ solution into the welfare function leads to the same function $W_3(y)$ as defined in the proof of \cref{prop:3}. It can be proven that $W_3(y)$ is guaranteed to be concave in the resource-led scaling scenario (thus the $d_0$ threshold in \cref{prop:3} is no longer needed). Meanwhile, we also plug the $x^*=f(c_f)$ solution into the welfare function and denote the resulting function by $W_4(y)$. Hence, $y^*=k$ is achieved when (a) $W_3'(k)\geq 0$, and (b) $W_3(k)\geq \max_{y\in [0,kf(c_f)^{1+\alpha}]}W_4(y)$. Condition (a) is equivalent to an upper bound on $g$ as shown in the proof of \cref{prop:3}. The same can be shown for condition (b): let $y_4'\doteq \arg\max_{y\in [0,kf(c_f)^{1+\alpha}]}W_4(y)$. Then condition (b) is equivalent to $g(k^\mu-(y_4')^\mu)$ being no greater than a certain term that involves no $g$. Since $k^\mu-(y_4')^\mu$ is lower bounded by $k^\mu-(kf(c_f))^\mu$, this condition can be translated into a well-defined upper bound on $g$. Setting $g_3$ as the minimum between this bound and the one that guarantees condition (a) proves the result for the case where $\frac{\theta}{k}\leq \frac{c_f(1+\alpha)}{\lambda}$. \hfill \emph{Q.E.D.}

\section{Additional Extensions and Robustness}
\label{ec:addextensions}

This section presents two further extensions that test the robustness of our main results to richer descriptions of how renewable capacity is financed, used, and delivered. We first consider decentralized renewable investment, replacing the policymaker's direct choice of capacity with an arrangement in which the policymaker sets a financial incentive and a utility makes the investment decision. We then model intermittent renewable availability and its firming explicitly: intermittency places a strictly positive floor under emissions whenever the developer is active, and firming restores the carbon-free outcomes at a higher marginal investment cost. Under decentralized investment, the adaptation-trap and adaptation-pathway characterizations of \cref{prop:1,prop:2} continue to hold, with the subsidy level determining how much clean capacity the private market delivers. Under intermittency, the trap mechanism is reinforced, whereas the pathway mechanism weakens, and firming restores the structure of the main model at a strictly higher marginal cost.

\subsection{Decentralized Renewable Investment}
\label{subsec:decentralized}

In this extension, we explicitly model the interaction between policy choices and renewable capacity investment. We focus on a setting in which the policymaker provides financial incentives (e.g., renewable subsidies and tax credits), a major policy lever for stimulating renewable investment in practice. We also consider a utility undertaking renewable capacity development, reflecting one of the primary channels through which renewable energy investment occurs in practice.

The utility determines the renewable capacity level $y$ to maximize the return from renewable investment. A stylized formulation of its objective is maximizing $(p_r-c_r+s)\cdot \beta(x^*(y),y) \cdot k(x^*(y))^{1+\alpha} - V(y)$, where $p_r$ is the unit rate (e.g., a large-load tariff) charged for the amount of renewable energy used by the developer under its optimal capability choices induced by $y$, $c_r$ is the unit renewable generation cost, and $s$ is the unit financial incentive that the policymaker provides for renewable generation.\footnote{Some renewable subsidies and tax benefits are tied to installed capacity rather than generation. This distinction does not affect the model outcomes because, in equilibrium, the developer either utilizes the full amount of available renewable generation or none of it.} Although utilities may also recover part of their renewable investment costs from other ratepayers in practice, we abstract from these channels and focus on the direct interaction between renewable investment and AI development.
In practice, electricity tariffs are typically subject to regulatory oversight and are therefore not readily adjusted, especially compared to direct financial incentives such as subsidies or tax credits. Hence, we treat $p_r$ as exogenous.%

We consider a three-stage game. In Stage 1, the policymaker chooses the financial incentive $s\in [0,\bar{s}]$, where $\bar{s}$ represents fiscal constraints that limit the amount of such incentive. The social welfare function that the policymaker maximizes equals the sum of the developer's profit (plus associated economic spillovers) and the utility's return from renewable development, minus the climate damage and the fiscal spending on incentive provision. In Stage 2, the utility determines the renewable capacity level $y$. Finally, in Stage 3, the developer chooses capability $x$ given the available renewable capacity. Relative to the main model, the developer pays the renewable tariff $p_r$ rather than $c_r$, although we continue to assume $p_r<c_f$, reflecting settings in which large-load renewable procurement provides electricity at a lower price than standard retail grid service.\footnote{A more rigorous formulation would introduce a separate fossil energy price $p_f>c_f$ representing the grid electricity tariff that the developer pays. This refinement only shifts the thresholds that define the equilibrium regimes and does not affect the insights qualitatively.  }

We show that our main results are robust to the policy-investment interactions, and the likelihood and the magnitude of the adaptation/pathway effects can be moderated by the intensity of the policy support. We start by presenting this result for the market-led scaling scenario. Recall the $y_1,y_2$ thresholds that define the developer's optimal capability choice in \cref{lem:1}. In the extension, we modify these thresholds by replacing $c_r$ with $p_r$ in their characterization (see the proof of \cref{lem:1}) and denote the new thresholds by $\hat{y}_1$ and $\hat{y}_2$.

\begin{proposition}
\label{prop:ext2}
\begin{enumerate}
        \item[(i)] Statement (i) in \cref{prop:1} continues to hold if any of the following conditions is satisfied: (a) $\theta\leq \max\{\frac{1+\alpha}{\lambda }\cdot c_fk, p_rk\}$, (b) $V'(k)\leq \min\{\eta (c_f-p_r)+(p_r-c_r)+(1-b) e_f \xi,p_r-c_r+\bar{s}\}$, or (c) $V'(\hat{y}_2)>p_r-c_r+\bar{s}$. 
        
        \item[(ii)] Otherwise, statement (ii) in \cref{prop:1} holds. Furthermore, the threshold $\bar{d}_m$ decreases in $\bar{s}$, whereas the equilibrium renewable energy investment $y^*$ when $d_0\geq \bar{d}_m$ increases in $\bar{s}$. \looseness=-1
        \end{enumerate}
\end{proposition}

\noindent \emph{Proof of \cref{prop:ext2}.} We can prove that condition (a) leads to statement (i) in \cref{prop:1} based on the same arguments in the proof of \cref{prop:1} (the second paragraph), because the developer's capability choice follows \cref{lem:1} with $c_r$ replaced by $p_r$. To prove condition (b), we complement the proof in \cref{prop:1} (the first paragraph) by the fact that $y^*=k$ can be achieved if and only if $V'(k)\leq p_r-c_r+\bar{s}$, that is, there exists an $s$ level to induce renewable investment to fully cover AI energy demand. Lastly, condition (c) indicates that under all feasible $s$ levels, the induced renewable investment does not exceed the $\hat{y}_2$ threshold, and thus the equilibrium remains in the carbon-free zone. 

To prove that statement (ii) in \cref{prop:1} continues to hold, we need to show that when the three conditions discussed above are violated, there are feasible $s$ values under which  $x^*=1$ can materialize. This is guaranteed by $V'(\hat{y}_2)\leq p_r-c_r+\bar{s}$. Furthermore, because $V'(k)> \min\{\eta (c_f-p_r)+(p_r-c_r)+(1-b) e_f \xi,p_r-c_r+\bar{s}\}$, either (i) $V'(k)> p_r-c_r+\bar{s}$, in which case no incentive level is sufficient to induce the utility to invest in $k$ amount of renewable capacity,  or (ii) $V'(k)> \eta (c_f-p_r)+(p_r-c_r)+(1-b) e_f \xi$, in which case even when $y=k$ is feasible, it is suboptimal (based on similar arguments in the fourth paragraph in the proof of \cref{prop:1}). Hence, we know that when $x^*=1$ is achieved, $y<k$ is guaranteed under the optimal $s$. Accordingly, let $s^*$ be the optimal incentive level, and $y^*(s^*)$ be the renewable investment induced. It can be shown that when $y^*(s^*)$ occurs on $W_2(y)$, which is defined in the proof of \cref{prop:1}, it equals either $y_2$ or $\hat{y}_2'\doteq \min\{y_2',(V')^{-1}(p_r-c_r+\bar{s})\}$. Replacing $y'_2$ with $\hat{y}_2'$ in the proof of \cref{prop:1} proves that statement (ii) continues to hold in this extension. Furthermore, we know that $\hat{y}_2'$ increases in $\bar{s}$, and so does  $W_2(\hat{y}_2')$. Therefore, the threshold $\bar{d}_m$ decreases in $\bar{s}$ according to its definition (with $y_2'$ replaced by $\hat{y}_2'$) in the proof of \cref{prop:1}. This completes the proof of the sensitivity results in \cref{prop:ext2}(ii).  \hfill \emph{Q.E.D.}

\cref{prop:ext2} indicates that policy support, when operating through the private investment incentives of renewable investors, may create a trade-off between the timing and carbon intensity of an adaptation trap. Although stronger policy support promotes renewable capacity investment and reduces the carbon intensity of AI development, it can simultaneously accelerate the attainment of frontier capability, thereby bringing forward the onset of the adaptation trap under the market-led scaling scenario. 

Nevertheless, we show that stronger policy support can uniformly reinforce the adaptation pathway %
in the resource-led scaling scenario.

\begin{proposition}
    \label{prop:ext3}
\begin{enumerate}
        \item[(i)] Statement (i) in \cref{prop:2} continues to hold if either (a) $\theta\geq \frac{1+\alpha}{\lambda}\cdot c_fk$ and $V'(k)> \min\{\eta(c_f-p_r)+(p_r-c_r)+(1-b) e_f \xi, p_r-c_r+\bar{s}\}$ or (b) $V'\left(kf(c_f)^{1+\alpha}\right)>p_r-c_r+\bar{s}$. 
        
        \item[(ii)] Otherwise, statement (ii) in \cref{prop:2} holds. The threshold $\bar{d}_c$ decreases in $\bar{s}$, whereas the equilibrium renewable energy investment $y^*$ when $d_0< \bar{d}_c$ is independent of $\bar{s}$. \looseness=-1
        \end{enumerate}
\end{proposition}

\noindent \emph{Proof of \cref{prop:ext3}.} To prove that statement (i) in \cref{prop:2} holds, we show that under either of the two conditions identified, only the first case of the developer's optimal capability choice $x^*$ identified in \cref{lem:2} can occur given a feasible incentive level; furthermore, the equilibrium renewable capacity investment $y^*<k$. This holds under condition (b) because it implies that the investment threshold $kf(c_f)^{1+\alpha}$ cannot be achieved by a feasible incentive level. To show the same under condition (a), we follow similar arguments in the first paragraph of the proof of \cref{prop:2}  if $V'(k)> \eta(c_f-p_r)+(p_r-c_r)+(1-b) e_f \xi$. Otherwise, we know that $V'(k)>p_r-c_r+\bar{s}$ must hold, which indicates that there does not exist a feasible incentive $s$ that can induce renewable capacity investment at level $k$, and thus $y^*<k$ must hold as well.  

To prove that statement (ii) in \cref{prop:2} holds when either condition is violated, we follow the proof of  \cref{prop:2} (from the second paragraph onward) by replacing $c_r$ with $p_r$. In particular, notice that because of the $\bar{s}$ upper bound, the definition of the $\bar{d}_c$ threshold should be changed to $\bar{d}_c=\min\{\bar{d}'(y), \forall y\in [kf(c_f)^{1+\alpha} ,(V')^{-1}(p_r-c_r+\bar{s})]\}$, which decreases in $\bar{s}$. However, the equilibrium renewable capacity investment before the adaptation pathway occurs, that is, the optimal $y^*$ that occurs on $W_1(y)$ defined in the proof of \cref{prop:2}, does not change as long as $V'\left(kf(c_f)^{1+\alpha}\right)\leq p_r-c_r+\bar{s}$, so that the entire  $W_1(y)$ piece is preserved in the extended model. \hfill \emph{Q.E.D.}

\subsection{Intermittency and Renewable Energy Utilization}
\label{subsec:utilization}

The main model assumes that the renewable capacity investment is the same as the usable renewable energy level. In practice, renewable
output is stochastic: solar produces nothing at night, and wind may exhibit multi-day lulls. This extension captures 
availability explicitly by modeling renewable output per unit of installed capacity as a random variable $A\in[0,1]$. The renewable energy used by the AI developer is then $\min\{Ay,\, kx^{1+\alpha}\}$, so the expected clean share of the developer's energy consumption equals 
\[
\tilde{\beta}(x,y)\doteq \mathbb{E}_A\left[\min\left\{\frac{Ay}{kx^{1+\alpha}},1\right\}\right]
\]
Replacing $\beta(x,y)$ in the main model by $\tilde{\beta}(x,y)$ leads to the extended model.

We focus on a simple two-point distribution in which $A=1$ with probability $p\in (0,1)$ and $A=0$ with probability $1-p$. This stylized distribution captures the key availability contrast between full and zero generation and offers tractability. We first characterize the developer's optimal capability choice and emission implications under intermittency. 

\begin{proposition}
\label{lem:intermittency}
    \begin{enumerate}
       \item[(i)] The $x^*(y)$ characterizations in \cref{lem:1,lem:2} continue to hold, with $c_r$ replaced by the blended cost $\tilde{c}_r\doteq pc_r+(1-p)c_f$ in their thresholds.
       \item[(ii)] In the market-led scaling scenario, $E(x^*(y),y)=0$ only if $y<y_1$. $E(x^*(y),y)$ increases in $y$ on $[y_1,y_2)$ and decreases in $y$ on $[y_2,k]$. However, for any $y'\in [y_1,y_2)$ and $y''\in [y_2,k]$, $E(x^*(y'),y')<E(x^*(y''),y'')$. 
        \item[(iii)] In the resource-led scaling scenario, $E(x^*(y),y)$ decreases in $y$ if $y<kf(c_f)^{1+\alpha}$, increases in $y$ on $[kf(c_f)^{1+\alpha}, kf(\tilde{c}_r)^{1+\alpha}]$, and is constant thereafter. Furthermore, when $\frac{1+\alpha}{\lambda}<\frac{c_f}{c_r}$ and $f(c_r)<1$, there exists a threshold $\bar{p}$ such that $E(x^*(k),k)>E(x^*(0),0)$ if $p\in (0,\bar{p})$. 
    \end{enumerate} 
\end{proposition}

\noindent \emph{Proof of \cref{lem:intermittency}.} To prove (i), note that $\tilde{\beta}(x,y)=p\cdot \min\left\{\frac{y}{kx^{1+\alpha}},1\right\} = p\cdot \beta(x,y)$, so the developer's expected energy cost equals $kx^{1+\alpha}\left[c_f-(c_f-c_r)\,p\cdot \beta(x,y)\right]=kx^{1+\alpha}\left[c_f-(c_f-(pc_r+(1-p)c_f))\beta(x,y)\right]$. This is the main model's energy cost with  $c_r$ replaced by $\tilde{c}_r\doteq pc_r+(1-p)c_f$. Because $\tilde{c}_r<c_f$ holds for $p\in(0,1)$, \cref{lem:1,lem:2} apply under this substitution.

To prove (ii), it can be calculated that under market-led scaling, $E(x^*(y),y)=E_1(y)\doteq e_f (1 - p) y$, which increases in $y$ when $y\in [y_1,y_2)$, and $E(x^*(y),y)=E_2(y)\doteq e_f k \left(1-\frac{p y}{k}\right)$, which decreases in $y$ when $y\in [y_2,k]$. Furthermore, it can be shown that for any $y'\in [y_1,y_2)$ and $y''\in [y_2,k]$, $E(x^*(y'),y')=E_1(y')<E_1(k)=E_2(k)\leq E_2(y'')= E(x^*(y''),y'')$. 
Similarly, to prove (iii), it can be calculated that under resource-led scaling, $E(x^*(y),y)=E_3(y)\doteq e_f \left(k \left(\min\{f(c_f),1\}\right)^{\alpha +1}-p y\right)$, which decreases in $y$ when $y<kf(c_f)^{1+\alpha}$. When $y\geq kf(c_f)^{1+\alpha}$, $E(x^*(y),y)=E_1(y)$ defined above if $y\leq  kf(\tilde{c}_r)^{1+\alpha}$ and stays constant at $E(x^*(y),y)=E_1(kf(\tilde{c}_r)^{1+\alpha})$. Finally, when $f(c_r)<1$, $f(c_f)<1$ must hold, and thus we have $E(x^*(0),0)=E_3(0)\doteq e_f\cdot k f(c_f)^{\alpha +1}$. Moreover, $f(\tilde{c}_r)<1$ also holds, so $E(x^*(k),k)=E_1(kf(\tilde{c}_r)^{1+\alpha})$. Denote $E_1(kf(\tilde{c}_r)^{1+\alpha})$ as a function $E_4(p)$. It can be calculated that $E_4(p)=E_3(0)$ when $p=0$. Additionally, %
when $c_r(1+\alpha)<c_f\lambda$,  $E_4'(p)>0$ if and only if $p<p_m\doteq \frac{c_f\lambda-c_r(1+\alpha)}{(c_f-c_r)\lambda}<1$, indicating that $E_4(p)$ increases on $(0,p_m)$ and decreases on $(p_m,1)$. Furthermore, $\lim_{p\uparrow 1}E_4(p)=0<E_4(0)$.
Hence, there must exist a threshold $\bar{p}\in (p_m,1)$ such that $E_4(p)>E_4(0)=E_3(0)$ if $p\in (0,\bar{p})$. This completes the proof of \cref{lem:intermittency}. \hfill \emph{Q.E.D.}

\cref{lem:intermittency} shows that intermittency tilts both scaling regimes toward carbon-intensive outcomes. Under market-led scaling, full decarbonization is unattainable whenever the developer is active. Even in the coupling zone in which renewable availability disciplines the developer's capability choice, renewable capacity fully covers AI's energy demand only when the resource is available; fossil energy is required otherwise. The emission-intensifying potential of renewable investment nevertheless remains visible. Once renewable capacity exceeds the threshold that induces frontier capability, fossil energy is used regardless of realized renewable output. Although additional renewable capacity reduces emissions within this frontier zone, the lowest emissions attainable there remain above the highest emissions generated in the coupling zone.

Under resource-led scaling, intermittency weakens, and in some cases reverses, the emission-mitigating role of renewable capacity. When compute requirements do not scale sufficiently faster than willingness to pay and intermittency is severe, greater renewable capacity can deepen fossil reliance and produce higher emissions than in the absence of renewables. The mechanism is that renewable capacity expands capability and can raise fossil use enough to offset, or even dominate, the clean-energy contribution. In that case, renewable expansion replicates its role under market-led scaling, functioning primarily as a capability-enabling channel rather than an emissions-reducing one. The following corollary characterizes when this resemblance arises for a given degree of intermittency.

\begin{corollary}
\label{cor:intermittency}
    In the resource-led scaling scenario, suppose $\frac{1+\alpha}{\lambda}<\frac{c_f}{c_r}$ and $f(c_r)<1$, as in \cref{lem:intermittency}(iii). Then for any availability probability $p\in (0,1)$, $E(x^*(k),k)>E(x^*(0),0)$ if $\frac{1+\alpha}{\lambda}<\frac{c_f(1-p)+ c_r p}{c_r}$.
\end{corollary}
\noindent \emph{Proof of \cref{cor:intermittency}.} This proof builds upon the proof of \cref{lem:intermittency}(iii). Specifically, we can show that the condition $p<\frac{c_f\lambda-c_r(1+\alpha)}{(c_f-c_r)\lambda}$ is equivalent to $\frac{1+\alpha}{\lambda}<\frac{c_f(1-p)+ c_r p}{c_r}$. Thus, $E_4$ is an increasing function from 0 to $p$, and $E_4(p)>E_3(0)$ is guaranteed.  \hfill \emph{Q.E.D.}

\cref{cor:intermittency} suggests that intermittency expands the regime in which renewable expansion may intensify emissions beyond the original market-led scaling scenario in the main analysis. Specifically, without intermittency, market-led and resource-led scaling are defined based on the comparison between $\lambda$ and $1+\alpha$ in the main analysis. The presence of intermittency attenuates the disciplining power of energy cost and shifts that comparison to $\lambda$ and $(1+\alpha)\cdot \frac{c_r}{c_rp+c_f(1-p)}$, which is smaller than $1+\alpha$ for any $p<1$. 

A direct implication of the above results is that intermittency strengthens the adaptation-trap mechanism but weakens the adaptation-pathway mechanism. We discuss this insight in detail next. First, we characterize the policymaker's renewable capacity decision. %

\begin{proposition}
\label{prop:intermittencycapacity1}
    In the market-led scaling scenario, the optimal renewable energy investment decision $y^*$ structurally follows \cref{prop:1}, with $c_r$ replaced by $c_rp+c_f(1-p)$ and $e_f$ replaced by $e_fp$ in its conditions. Specifically, $y^*<y_2$ if condition (a) in \cref{prop:1}(i) holds under this substitution; $y^*=k$ if condition (b) holds under this substitution; and otherwise there exists a threshold in $d_0$ beyond which $y^*\geq y_2$, paralleling \cref{prop:1}(ii).  
\end{proposition}
\noindent \emph{Proof of \cref{prop:intermittencycapacity1}.} The characterization follows directly from the substitution $c_r\mapsto\tilde{c}_r$ established in the proof of \cref{lem:intermittency}(i), together with the calculation of expected AI emissions in the frontier zone: under intermittency, each unit of renewable capacity reduces emissions by $e_fp$ instead of $e_f$. \hfill \emph{Q.E.D.}

\cref{prop:intermittencycapacity1} shows that intermittency reinforces the adaptation trap:  condition (b) in \cref{prop:1}(i) becomes more stringent when $c_r$ is replaced by the higher blended cost $\tilde{c}_r$ and $e_f$ is reduced to $e_fp$. Consequently, a sufficiently high baseline emissions stock $d_0$ (the condition in \cref{prop:1}) is more likely to drive the equilibrium into the frontier zone ($y^*\geq y_2$), in which emissions are higher than in the coupling zone according to \cref{lem:intermittency}(ii). This transition causes $d_0$ to deteriorate more rapidly and makes the trap increasingly difficult to escape. Intermittency further extends the risk of an adaptation trap to the coupling zone: even an equilibrium in this zone can worsen $d_0$ and eventually move the equilibrium into the higher-emission frontier zone.

Whether greater intermittency, represented by a lower $p$, accelerates this transition is generally ambiguous because two opposing forces are at work. On the one hand, intermittency erodes the cost advantage of renewables and weakens their ability to relax scaling constraints, thereby delaying entry into the frontier zone. On the other hand, the same weakening may induce the policymaker to invest more aggressively in renewable capacity when the incentive to support frontier AI development is sufficiently strong. Numerical studies suggest that the former/latter force is likely to dominate when $p$ is large/small. 

We derive a similar result for the resource-led scaling scenario to show that the optimal capacity investment decision structurally follows \cref{prop:2}. Namely, the equilibrium enters the coupling zone in which $y^*\geq kf(c_f)^{1+\alpha}$. However, because the coupling zone may lead to higher emissions due to intermittency according to
\cref{lem:intermittency}(iii) and \cref{cor:intermittency}, this transition may amount to an adaptation trap rather than a pathway to alleviate carbon dependence.

\textbf{Renewable Energy Firming:} In practice, renewable intermittency can be mitigated through firming measures such as energy storage and demand response. We capture these measures through a stylized firming technology that converts the two-point renewable output into a constant, dependable supply of $\zeta\leq p$ units of energy per unit of installed capacity. (We assume $\zeta\leq p$ because firming cannot deliver more energy on average than the resource generates.)
We refer to $\zeta$ as the firmed utilization rate. It summarizes the effectiveness of the firming system, including storage capacity, round-trip efficiency, and the extent to which demand can be shifted to coincide with renewable availability.

Given $\zeta$, the share of the developer's energy demand met by renewables equals $\hat{\beta}(x,y)=\min\left\{\frac{\zeta y}{kx^{1+\alpha}},1\right\}$. Replacing $\beta(x,y)$ by $\hat{\beta}(x,y)$ leads to the extended model with firming. 
We solve this model and show that our main results continue to hold. Proposition \ref{prop:ext1} formalizes the findings.  

\begin{proposition}
\label{prop:ext1}
Assume $p=1$. 
 \begin{enumerate}
        \item[(i)]Statement (i) of \cref{prop:1} and \cref{prop:2} continues to hold once $V'(k)$ is replaced by $\frac{1}{\zeta}\cdot V'(\frac{k}{\zeta})$. 

  \item[(ii)]  Both the adaptation trap threshold $\bar{d}_m$ in \cref{prop:1} and the adaptation pathway threshold $\bar{d}_c$ in \cref{prop:2} decrease in $\zeta$. 
   \end{enumerate}
\end{proposition}

\noindent \emph{Proof of \cref{prop:ext1}.}
The renewables share function with firming $\hat{\beta}(x,y)$ is equivalent to $\beta(x,\zeta y)$. Hence, by treating $\zeta y$ as one variable $z$, the modified model is equivalent to the base model (with decision variables $\{x,z\}$) except that the renewable investment cost function $V(y)$ should be replaced by $\hat{V}(z)\doteq V(\frac{z}{\zeta})$. Hence, the proofs of \cref{prop:1} and \cref{prop:2} continue to apply to the extended model with the decision variable $y$ replaced by $z$, the corresponding threshold notation and expressions adjusted accordingly, and $V(y)$ changed to $\hat{V}(z)$. In particular, we change $V'(k)$ to $\hat{V}'(k)=\partial V(\frac{z}{\zeta})/\partial z|_{z=k}=\frac{1}{\zeta} \cdot V'(\frac{k}{\zeta})$. This completes the proof of \cref{prop:ext1}(i).

To prove \cref{prop:ext1}(ii), we first analyze the sensitivity of the adaptation trap threshold $\bar{d}_m$ in \cref{prop:1}. A key finding is that $\hat{V}(z)-\hat{V}(z'_2)=V(\frac{z}{\zeta})-V(\frac{z'_2}{\zeta})$ increases in $\zeta$ for all $z<z_2$ ($z'_2$ and $z_2$ defined in the same way as $y'_2$ and $y_2$ in the proof of \cref{prop:1}). Specifically, given the convexity of the $V$ function, it can be calculated that $\partial \left(V(\frac{z}{\zeta})-V(\frac{z'_2}{\zeta})\right)/\partial \zeta= \frac{1}{\zeta^2}\left(z'_2V'(\frac{z'_2}{\zeta})-zV'(\frac{z}{\zeta})\right)>0$. Therefore, $\Delta W(z)=W(z)-W_2(z_2')$, which is the counterpart of the $\Delta W(y)$ function defined in the proof of \cref{prop:1} and equals a term independent of $\zeta$ minus $\hat{V}(z)-\hat{V}(z'_2)$, decreases in $\zeta$. Consequently, the $\bar{d}''(z)$ threshold, which is the counterpart of $\bar{d}''(y)$ in the proof of \cref{prop:1},  decreases in $\zeta$ as well, and so does $\bar{d}_m$, which is the supremum of $\bar{d}''(z)$. Based on similar arguments, we show that $\Delta W(z)$, as the counterpart of the $\Delta W(y)$ function defined in the proof of \cref{prop:2}, increases in $\zeta$. Hence, the adaptation pathway threshold $\bar{d}_c$ in \cref{prop:2} decreases in $\zeta$ according to how it is defined in the proof of \cref{prop:2}.  \hfill \emph{Q.E.D.}

\cref{prop:ext1}(i) suggests that although firming can restore the carbon-free outcome in the coupling zone, it is less likely to happen than in the main model (which is a special case with $\zeta=1$), because imperfect energy utilization increases the marginal investment cost from $V'(k)$ to $\hat{V}'(k)$. In fact, comparing the coupling condition with and without firming (presented in \cref{prop:ext1}(i) and \cref{prop:intermittencycapacity1}, respectively) shows that coupling tends to be more challenging to achieve with firming due to the convexity of the investment cost. The trade-off is that firming mitigates emissions once in the coupling zone.   Taken together, \cref{lem:intermittency}(ii) and \cref{prop:ext1} point to a common conclusion: without firming, full decarbonization is infeasible; with firming (compared to a dependable utilization rate, $p=1$, as a deterministic benchmark), it must clear a strictly higher marginal-cost hurdle.

Nevertheless, the moderating effect of the utilization rate $\zeta$ could differ across the market-led and resource-led scaling regimes. Proposition \ref{prop:ext1} indicates that a higher $\zeta$ uniformly facilitates the adaptation pathway in the resource-led scaling regime: it not only enlarges the parametric region in which an adaptation pathway exists (by decreasing $\frac{1}{\zeta}V'(\frac{k}{\zeta})$ so the condition in \cref{prop:2}(i) is less likely to hold), but also accelerates the realization of such a pathway (by reducing the threshold $\bar{d}_c$). However, a higher $\zeta$ may not be universally beneficial in the market-led scaling regime. Notably, although a sufficiently large $\zeta$ reduces the likelihood of an adaptation trap (by making the condition in \cref{prop:1}(i) more likely to hold), when $\zeta$ is not large enough to eliminate the possibility of a trap altogether, an increase in $\zeta$ could conversely speed up the emergence of the trap (by reducing the threshold $\bar{d}_m$). Although directionally distinct, the effects in the two scaling regimes stem from the same mechanism: more effective renewable utilization incentivizes more investment, thereby driving the equilibrium into the carbon-intensive zone in the market-led scaling regime or the carbon-free zone in the resource-led scaling regime. %

\section{Technical Details for the Extensions in the Main Text}
\label{ec:techdetails}

This section provides formal statements and supporting derivations for the extensions discussed in the main text. We first characterize how AI-enabled reductions in renewable investment costs affect the conditions for full decarbonization. We then present the model formulation for the competition extension and outline the solution approach. \looseness=-1

\subsection{AI-Enabled Reductions in Clean-Capacity Costs}
\label{subsec:td-costreduction}

In this subsection, we formalize the finding that an AI-enabled reduction in renewable investment costs relaxes the cost requirement for achieving full decarbonization. 

\begin{proposition}
    Given the renewable investment function $\phi(x)V(y)$ where $V(y)=g\cdot y^\mu$, \cref{prop:3} continues to hold if the thresholds $g_1$ and $g_2$ are replaced by $\frac{g_1}{\phi(1)}$ and $\frac{g_2}{\phi(1)}$.
\end{proposition}

The proof follows that of \cref{prop:3} except for replacing $g$ with $\phi(1)\cdot g$. 

\subsection{Competition Among AI Developers}
\label{subsec:td-competition}
We provide a detailed model formulation and technical details of the competition extension. We consider two symmetric developers that engage in capability and price competition by choosing $(x_1,x_2)$ first and then prices $(p_1,p_2)$. 

Note that this setting extends the main model in two ways: first, the number of developers is increased; second, an additional pricing decision is introduced. To isolate the effects, we first show that extending the main model into a price-setting monopoly does not change the results. To meaningfully model the price decision, we introduce market heterogeneity by assuming the  market willingness-to-pay parameter $\theta'$ is uniformly distributed on $[0,\tilde{\theta}]$. The total demand is normalized to 1. Let $p$ denote the monopolistic developer's price decision. Accordingly, customers invest if $\theta' x^\lambda\geq p$, indicating that the market demand equals $\frac{\tilde{\theta}-px^{-\lambda}}{\tilde{\theta}}$, and the developer's profit function becomes $\tilde{\Pi}(x,y,p)=p\frac{\tilde{\theta}-px^{-\lambda}}{\tilde{\theta}}-\Big(\beta(x,y)c_r+\big(1-\beta(x,y)\big)c_f\Big)\,k\,x^{1+\alpha}.$ 
Optimizing this function over the price decision leads to $\tilde{\Pi}^*(x,y)=\max_{p\geq 0}\tilde{\Pi}(x,y,p)=\frac{\tilde{\theta}}{4}x^\lambda-\Big(\beta(x,y)c_r+\big(1-\beta(x,y)\big)c_f\Big)\,k\,x^{1+\alpha}$. Observe that $\tilde{\Pi}^*(x,y)$ is the same as the developer's profit function $\Pi(x,y)$ formulated in the main model, except that $\theta$ is replaced by $\tilde{\theta}/4$. Therefore, we shall set $\tilde{\theta}=4\theta$ in this extension to equate the equilibrium outcome from the price-setting monopoly and the main model so as to isolate the impact of competition. 

We are now ready to formulate the duopoly extension. Let $y_i$ be the amount of renewable investment capacity that developer $i=1,2$ has access to. First consider the pricing stage: given $(x_1,x_2)$ (where $x_1>x_2$ without loss of generality) and $(p_1,p_2)$, customers invest in $x_1$ if their willingness to pay satisfies $\theta' x_1^{\lambda }-p_1\geq \theta'  x_2^{\lambda }-p_2$, which leads to $\theta'\geq \frac{p_1-p_2}{x_1^{\lambda }-x_2^{\lambda }}$. Furthermore, for customers to make an investment, $\theta' x_2^{\lambda}-p_2\geq 0$ should be satisfied, which gives $\theta'\geq p_2 x_2^{-\lambda }$. Hence, let $C(x,y)\doteq \Big(\beta(x,y)c_r+\big(1-\beta(x,y)\big)c_f\Big)\,k\,x^{1+\alpha}$ denote the energy cost under a given $(x,y)$; then developer 1's profit function equals $\Pi_1(x_1,y_1,p_1,p_2) = p_1 \left(\tilde{\theta}-\frac{p_1-p_2}{x_1^{\lambda }-x_2^{\lambda }}\right)/{\tilde{\theta}}-C(x_1,y_1)$. Developer 2's profit function equals $\Pi_2(x_2,y_2,p_1,p_2)=p_2 \left(\frac{p_1-p_2}{x_1^{\lambda }-x_2^{\lambda }}-p_2 x_2^{-\lambda }\right)/{\tilde{\theta}}-C(x_2,y_2)$. Maximizing these two profit functions together leads to a price equilibrium: $p_1^*=\frac{2  x_1^{\lambda } \left(x_1^{\lambda }-x_2^{\lambda }\right)\tilde{\theta}}{4 x_1^{\lambda }-x_2^{\lambda }}$, and $p_2^*= \frac{x_2^{\lambda } \left(x_1^{\lambda }-x_2^{\lambda }\right)\tilde{\theta}}{4 x_1^{\lambda }-x_2^{\lambda }}$. We plug the price equilibrium into the developers' profit functions and analyze the capability-choice equilibrium. This part of the analysis becomes analytically intractable. We resort to an extensive numerical study, which leads to the results shown in \cref{subsec:competition}.

\section{Proofs of Analytical Results}
\label{ec:proofs}

For convenience, we rewrite the developer's profit in equation (2) as follows: Let $\pi(x)\doteq \theta x^\lambda - c\cdot (kx^{1+\alpha})$. The developer's profit equals $\pi(x)$ where $c=c_r$ if $kx^{1+\alpha}\leq y$ (i.e., AI energy demand can be fully covered by renewable sources). We denote this piece by $\pi_{c_r}(x)$. Otherwise, it equals $\pi(x)$ where $c=\beta c_r+(1-\beta)c_f$, where $\beta=\frac{y}{kx^{1+\alpha}}$; this function can be rewritten as $\theta x^\lambda - c_f\cdot (kx^{1+\alpha})+(c_f-c_r)y$, which we denote by $\pi_{c_f}(x)$. \looseness=-1

\vspace{0.3cm}

\noindent \emph{Proof of \cref{lem:1}.} 
The special case of this proposition in which $\lambda=1+\alpha$ can be proven by setting $y_1=y_2=0$ if $\theta\geq c_fk$, setting $y_1=y_2=k+1$ if $\theta\leq c_rk$, and setting $y_1=0, y_2=k+1$ if $\theta\in (c_rk,c_fk)$.

To prove this proposition for the case in which $\lambda>1+\alpha$, we first introduce two thresholds on $\theta$: The first threshold is defined as $\bar{\theta}\doteq \frac{(\alpha +1) c_f k}{\lambda }$. The second threshold, denoted by $\bar{\bar{\theta}}$, is defined as follows: Consider the function $F(\theta)=-\left((c_f-c_r) k^{\frac{\lambda }{-\alpha +\lambda -1}} \left(\frac{c_r}{\theta }\right)^{\frac{\alpha +1}{-\alpha +\lambda -1}}\right)+c_f k-\theta$. It can be shown that $F''(\theta)<0$, indicating that $F$ is concave and thus can have at most two zeros. It can be verified that one of them is $c_rk$ because $F(c_rk)=0$. Denote the other zero, if it exists, by $\bar{\bar{\theta}}$.

We start by proving \cref{lem:1} for the case in which $\alpha=0$. In this case, $\pi''(x)\geq 0$ under the assumption that $\lambda>1+\alpha=1$, indicating that $\pi(x)$ is convex and achieves its maximum at the boundary points. As such, we analyze two sub-cases. 
\begin{itemize}
\item The first sub-case is one in which $\frac{c_r}{c_f}<\frac{1+\alpha}{\lambda}$ and $\theta\in \left(\bar{\theta},\bar{\bar{\theta}}\right)$. Define $y^0_1\doteq k^{\frac{\lambda }{-\alpha +\lambda -1}} \left(\frac{c_r}{\theta }\right)^{\frac{\alpha +1}{-\alpha +\lambda -1}}$. We also introduce another threshold $y^0_2$: Consider the function $G(y)=\theta  \left(\left(\frac{y}{k}\right)^{\frac{\lambda }{\alpha +1}}-1\right)-c_f k \left(\frac{y}{k}-1\right)$. It can be shown that $G''(y)>0$, indicating that the function is convex. Furthermore, it can be verified that $G(k)=0$ and $G'(k)>0$ when $\theta>\bar{\theta}$. This implies that in the first sub-case, $G(y)$ has another zero that is strictly smaller than $k$; denote that zero by $y^0_2$. Hence, $G(y)<0$ if and only if $y\in (y^0_2,k)$. 

We further prove that $y^0_1<y^0_2$. We prove this by showing that $G(y^0_1)>0$. It can be shown that $G(y^0_1)=F(\theta)$. We show that $\forall \theta\in (\bar{\theta},\bar{\bar{\theta}})$, $F(\theta)>0$ is satisfied. We show this in two steps. First, we show that $F(\theta)>0$ $\forall \theta\in(c_rk,\bar{\bar{\theta}})$. This is because when $\frac{c_r}{c_f}<\frac{1+\alpha}{\lambda}$, it can be calculated that $F'(c_rk)>0$ and $\lim_{\theta\rightarrow\infty}F(\theta)<0$, indicating that $\bar{\bar{\theta}}$ exists and is greater than $c_rk$, and $F(\theta)>0$ if and only if $\theta\in(c_rk,\bar{\bar{\theta}})$. Second, we show that $c_rk<\bar{\theta}<\bar{\bar{\theta}}$ and thus $(\bar{\theta},\bar{\bar{\theta}})\subset (c_rk,\bar{\bar{\theta}})$. The inequality $c_rk<\bar{\theta}$ directly follows from the condition that $\frac{c_r}{c_f}<\frac{1+\alpha}{\lambda}$.  To show $\bar{\theta}<\bar{\bar{\theta}}$, we analyze $F(\bar{\theta})$. It can be shown that $F(\bar{\theta})=0$ when $c_r=\frac{c_f(1+\alpha)}{\lambda}$; furthermore, it can be calculated that $F(\bar{\theta})$ is strictly decreasing in $c_r$ given $\frac{c_r}{c_f}<\frac{1+\alpha}{\lambda}$ and $\lambda>1+\alpha$. This indicates that $F(\bar{\theta})>0$ and, in turn, $\bar{\theta}<\bar{\bar{\theta}}$ holds, $\forall c_r<\frac{c_f(1+\alpha)}{\lambda}$. This completes the proof that $y^0_1<y^0_2$ in this sub-case. Accordingly, we can show the following.
\begin{itemize}
    \item When $y<y^0_1$, it can be shown that (i) $\pi_{c_r}(0)>\pi_{c_r}\left((y/k)^{1/(1+\alpha)}\right)$, and (ii) $\pi_{c_f}(1)<\pi_{c_f}\left((y/k)^{1/(1+\alpha)}\right)$, indicating that the optimal technology choice is $x=0$ due to the convexity of the functions $\pi_{c_r}$ and $\pi_{c_f}$. To prove (i), consider the function $J(y)=\pi_{c_r}\left((y/k)^{1/(1+\alpha)}\right)-\pi_{c_r}(0)$. It can be shown that $J''(y)>0$ and that $J(y)$ has two zeros, $y=0$ and $y=y^0_1$. Hence, if and only if $y\in (0,y^0_1)$, $J(y)<0$, and thus (i) holds. To prove (ii), note that $\pi_{c_f}\left((y/k)^{1/(1+\alpha)}\right)-\pi_{c_f}(1)$ is the $G(y)$ function introduced above. Because $G(y)>0$ $\forall y<y^0_2$ and $y^0_1<y^0_2$, (ii) is proven. 
    \item Based on the above observations, we know that when $y\in [y^0_1,y^0_2)$, (i) $\pi_{c_r}(0)\leq \pi_{c_r}\left((y/k)^{1/(1+\alpha)}\right)$, and (ii) $\pi_{c_f}(1)<\pi_{c_f}\left((y/k)^{1/(1+\alpha)}\right)$, indicating that the optimal technology choice is $x=(y/k)^{1/(1+\alpha)}$.  
    \item Similarly, when $y\geq y^0_2$, we have (i) $\pi_{c_r}(0)< \pi_{c_r}\left((y/k)^{1/(1+\alpha)}\right)$, and (ii) $\pi_{c_f}(1)\geq \pi_{c_f}\left((y/k)^{1/(1+\alpha)}\right)$, indicating that the optimal technology choice is $x=1$.  
\end{itemize}
Hence, setting $y_1=y_1^0$ and $y_2=y_2^0$ proves \cref{lem:1} in the first sub-case.

\item We then consider the second sub-case in which the condition in the first sub-case is not met. 
\begin{itemize}
\item (A) If $\frac{c_r}{c_f}<\frac{1+\alpha}{\lambda}$ and $\theta\leq \bar{\theta}$, then $G'(k)\leq 0$, indicating that $G(y)\geq 0$ and, accordingly, $\pi_{c_f}\left((y/k)^{1/(1+\alpha)}\right)\geq \pi_{c_f}(1)$ holds $\forall y\in [0,k]$. Hence, setting $y_1=y_1^0$ and $y_2=k$ proves the proposition. \looseness=-1

\item (B) If $\frac{c_r}{c_f}<\frac{1+\alpha}{\lambda}$  and $\theta\geq \bar{\bar{\theta}}$, then $G(y^0_1)=F(\theta)\leq 0$, which indicates that 
 $y^0_1\in (y^0_2,k)$. 
In that case, the solution $x=(y/k)^{1/(1+\alpha)}$ is suboptimal. Hence, setting $y_1=y_2=\frac{c_f k-\theta }{c_f-c_r}$ proves the proposition, because it can be shown that $\pi_{c_r}(0)>\pi_{c_f}(1)$ if and only if $y<\frac{c_f k-\theta }{c_f-c_r}$.

\item (C) If $\frac{c_r}{c_f}\geq \frac{1+\alpha}{\lambda}$, it can be shown that $\bar{\bar{\theta}}\leq c_rk$, and thus $F(\theta)>0$ if and only if $\theta\in(\bar{\bar{\theta}},c_rk)$.  
\begin{itemize}
\item (C-i) If $\theta\geq c_rk$, then $F(\theta)\leq 0$, that is, $G(y^0_1)\leq 0$, which is the situation discussed in (B).
\item (C-ii) If $\theta<c_rk$, it can be shown that $y_1^0>k$, indicating that $J(y)\leq 0$ and, accordingly, $\pi_{c_r}(0)\geq \pi_{c_r}\left((y/k)^{1/(1+\alpha)}\right)$ $\forall y\in [0,k]$. Hence, the solution $x=(y/k)^{1/(1+\alpha)}$ is also suboptimal. Hence, applying the approach in (B) and setting $y_1=y_2=\frac{c_f k-\theta }{c_f-c_r}$ proves the proposition.
\end{itemize}
\end{itemize}
\end{itemize}

We then analyze the case in which $\alpha>0$. It can be shown that $\pi''(x)<0$ if  $x\leq \left(\frac{\theta  (\lambda -1) \lambda }{\alpha  (\alpha +1) ck}\right)^{\frac{1}{\alpha -\lambda +1}}$, and $\pi''(x)\geq 0$ otherwise, indicating that $\pi(x)$ is concave first and then becomes convex as $x$ increases. %
Furthermore, it can be calculated that $\pi'(x)<0$ for $x$ sufficiently small (because $\lambda-1>\alpha$ in this case). These observations indicate that $\pi(x)$ continues to achieve its maximum at the boundary points, and the arguments used to analyze the case in which $\alpha=0$ continue to apply.  \hfill \emph{Q.E.D.} 

\vspace{0.3cm}

\noindent \emph{Proof of \cref{lem:2}.} 
First, consider the case in which $\lambda\leq 1$. In this case, $\pi''(x)=\theta  (\lambda -1) \lambda  x^{\lambda -2}-\alpha  (\alpha +1) ck x^{\alpha -1}\leq 0$, indicating that the function $\pi(x)$ is concave. Hence, the global maximum of $\pi(x)$ is achieved at its stationary point. %
It can be shown that when  $y> k \left(\frac{\theta  \lambda }{(1+\alpha)  c_rk}\right)^{\frac{1+\alpha}{1+\alpha -\lambda}}$, the first piece $\pi_{c_r}(x)$ achieves its global maximum (if it is less than 1), and the second piece $\pi_{c_f}(x)$ decreases. For any $y\in \left[k \left(\frac{\theta  \lambda }{(1+\alpha)  c_fk}\right)^{\frac{1+\alpha}{1+\alpha -\lambda}},k \left(\frac{\theta  \lambda }{(1+\alpha)  c_rk}\right)^{\frac{1+\alpha}{1+\alpha -\lambda}}  \right]$, the first/second piece increases/decreases, and thus the maximum is achieved at the break point at which $x=\left(\frac{y}{k}\right)^{\frac{1}{1+\alpha}}$. For any $y<k \left(\frac{\theta  \lambda }{(1+\alpha)  c_fk}\right)^{\frac{1+\alpha}{1+\alpha -\lambda}}$, the first piece increases, and the second piece achieves its global maximum (if it is less than 1). %
Incorporating the boundary constraint that $x\leq 1$ leads to \cref{lem:2} for this case. 

We then consider the case in which $\lambda>1$. In this case, it can be shown that $\pi''(x)\geq 0$ if $x<\left(\frac{\lambda(\lambda-1)\theta}{\alpha(1+\alpha)ck}\right)^{\frac{1}{1+\alpha-\lambda}}$ and 
$\pi''(x)< 0$  otherwise, indicating that $\pi$ is convex first and then becomes concave as $x$ increases.  
Furthermore, $\pi'(x)=\theta  \lambda  x^{\lambda -1}-(\alpha +1) ck x^{\alpha }>0$ for $x$ sufficiently small (because $\lambda-1<\alpha$ in this case), indicating that $x=0$ cannot be the maximum. Therefore, $\pi(x)$ continues to achieve its global maximum at a stationary point, and thus the arguments used in the previous case continue to apply. \hfill \emph{Q.E.D.}

\vspace{0.3cm}

\noindent \emph{Proof of \cref{prop:1}.} Recall the $y_1$ and $y_2$ thresholds identified in \cref{lem:1}. When $V'(k)\leq \eta (c_f-c_r)+(1-b) e_f \xi$, it can be shown that either $y^*<y_2$ or $y^*=k$, in which case $k(x^*(y^*))^{1+\alpha}\leq y^*$ is satisfied, indicating that $E(x^*(y^*),y^*)=0$ always holds. To show this, assume that $y^*\in (y_2,k)$. Then, according to \cref{lem:1}, $x^*(y^*)=1$. Plugging this solution into the policymaker's welfare function, we obtain $W(y^*)=\eta\left(-c_f(k-y^*)-c_ry^*+\theta\right)-\xi(1-b)(d_0+e_f(k-y^*))-V(y^*)$. It can be calculated that $W'(y^*)=\eta(c_f-c_r)+\xi(1-b)e_f-V'(y^*)$. Because $V(y)$ is convex and $y^*<k$, we have $V'(y^*)<V'(k)\leq \eta (c_f-c_r)+(1-b) e_f \xi$, indicating that $W'(y^*)>0$, which contradicts the optimality of $y^*$ as an interior solution.

When $\theta\leq \max\{\frac{1+\alpha}{\lambda }\cdot c_fk, c_rk\}$, two sub-cases arise: (i) $\frac{c_r}{c_f}<\frac{1+\alpha}{\lambda}$ and $\theta\leq \frac{1+\alpha}{\lambda }c_fk$, which is the $\bar{\theta}$ threshold defined in the proof of \cref{lem:1}, and (ii) $\frac{c_r}{c_f}\geq \frac{1+\alpha}{\lambda}$ and $\theta\leq c_rk$.
Case (i) corresponds to situation (A) discussed in the second sub-case in the proof of \cref{lem:1}, in which $y_2=k$ and thus the third scenario of $x^*$ in \cref{lem:1} does not materialize. Hence, $k(x^*(y^*))^{1+\alpha}\leq y^*$ always holds. 
Case (ii) corresponds to situation (C-ii) discussed in the proof of \cref{lem:1}. In that case, it can be shown that because $\theta<c_rk$, $y_1=y_2=\frac{c_fk-\theta}{c_f-c_r}>k$, indicating that $x^*=0$ $\forall y\in [0,k]$. Hence, $k(x^*(y^*))^{1+\alpha}\leq y^*$ always holds as well. \looseness=-1

Combining the above two paragraphs proves \cref{prop:1}(i). 

Next, consider the case in which neither condition in \cref{prop:1}(i) holds. This corresponds to either the first sub-case, or situations (B) and (C-i) in the second sub-case discussed in the proof of \cref{lem:1}. In all of these cases, the third scenario of $x^*$ in \cref{lem:1} can materialize when $y$ is sufficiently large, in particular, when $y=k$. Hence, $W(k)$, that is, the welfare function when $y=k$, should be calculated by plugging in $x^*(k)=1$. Based on the calculation in the first paragraph of this proof, we know that when $V'(k)> \eta (c_f-c_r)+(1-b) e_f \xi$, $W'(k)<0$, implying that $y^*<k$. If $y_2\leq 0$, then according to \cref{lem:1}, $k(x^*(y))^{1+\alpha}>y$ holds for all $y\in[0,k)$.  Setting $\bar{d}_m$ to be any negative number proves \cref{prop:1}(ii). 

If $y_2>0$, then $W(y)$ can be formulated as a piecewise function containing at least two pieces: 
The last piece corresponds to the $y$ range of $[y_2,k]$, in which $x^*(y)=1$; we denote this piece by $W_2$. At least one other piece corresponds to the first or second case of the  $x^*$ solution identified in \cref{lem:1}. %
We show that $y^*$ occurs on $W_2(y)$ if and only if $d_0$ is sufficiently large. It can be calculated that  $W_2''(y)<0$, indicating that $W_2(y)$ is a concave function. Hence, if $y^*$ occurs on $W_2(y)$, then $y^*$ equals either the stationary point at which $W_2'(y)=0$ (denoted by $y'_2$) or the endpoint $y_2$ (because $y^*<k$, as shown in the previous paragraph). 
\begin{itemize}
    \item If $\eta(c_f-c_r)+(1-b)e_f \xi-V'(y_2)>0$ (a condition independent of $d_0$ because $y_2$ is the zero of a function that is independent of $d_0$; see the proof of \cref{lem:1}), then $W_2(y'_2)>W_2(y_2)$. Hence, $y^*$ occurs on $W_2(y)$ if and only if $W(y)\leq W_2(y_2')$ holds for any $y< y_2$. It can be calculated that for any $y
    \leq y_2$,  $\frac{\partial(W(y)- W_2(y_2'))}{\partial d_0}$ equals either $-b \xi$ (if the first case of $x^*$ in \cref{lem:1} materializes under $y$, i.e., $x^*(y)=0$) or $-b \xi  \left(1-\left(\frac{y}{k}\right)^{\frac{1}{\alpha +1}}\right)$ (if the second case of $x^*$ in \cref{lem:1} materializes, i.e., $x^*(y)=(\frac{y}{k})^{1/(1+\alpha)}$). Because condition (a) in \cref{prop:1}(i) does not hold, we know $y_2<k$ according to the proof of \cref{lem:1} (see the first sub-case, and situations (B) and (C-i) in the second sub-case). %
    Hence, $y<y_2$ implies that $y<k$. Consequently, both derivatives calculated above are negative for any $y<y_2$ given $b>0$ (and are independent of $d_0$), indicating that  $W(y)- W_2(y_2')$ is a strictly decreasing linear function in $d_0$. Hence, if $y^*$ occurs on $W_2(y)$ under some $d'$, then the same happens under any $d''>d'$, indicating that either \cref{prop:1}(i) or (ii) holds. To show that \cref{prop:1}(ii) holds in this case, define $\Delta W(y)=W(y)- W_2(y_2')$ when $d_0=0$ for any $y<y_2$. %
Define $\bar{d}''(y)=\max\{0,\Delta W(y)/|\frac{\partial(W(y)- W_2(y_2'))}{\partial d_0}|\}$, which is finite for any $y<y_2<k$. By definition, $W(y)<W_2(y_2')$ when $d_0>\bar{d}''(y)$. Let $\bar{\bar{d}}''=\sup_{y<y_2}\bar{d}''(y)$. Then, when $d_0>\bar{\bar{d}}''$, $W(y)<W_2(y_2')$ $\forall y<y_2$, and thus $y^*=y_2'$, indicating that \cref{prop:1}(ii) holds in this sub-case if we set $\bar{d}_m=\bar{\bar{d}}''$.

\item If $\eta(c_f-c_r)+(1-b)e_f \xi-V'(y_2)\leq 0$, then $W_2(y'_2)\leq W_2(y_2)$. Replacing $y_2'$ by $y_2$ in the arguments used in the previous bullet point proves \cref{prop:1}(ii) in this sub-case. \hfill \emph{Q.E.D.}

\end{itemize}

\vspace{0.3cm}

\noindent \emph{Proof of \cref{prop:2}.}
First, consider the conditions in \cref{prop:2}(i): $\theta\geq \frac{1+\alpha}{\lambda}c_fk$ indicates that $k \left(\frac{\theta  \lambda }{(1+\alpha)  c_fk}\right)^{\frac{1+\alpha}{1+\alpha -\lambda}}\geq k$ (i.e., $f(c_f)\geq 1$) given $\lambda<1+\alpha$. In that case, $x^*(y)=1$ $\forall y\in [0,k]$ according to \cref{lem:2}. Plugging $x^*(y)=1$ into the welfare function leads to a function that we denote by $\bar{W}_1(y)$. It can be calculated that  $\bar{W}_1''(y)<0$, indicating that $\bar{W}_1(y)$ is a concave function. It can be calculated that when $V'(k)>\eta(c_f-c_r)+(1-b) e_f \xi$, $\bar{W}_1'(k)<0$ and thus $y^*<k$. Hence, $k(x^*(y))^{1+\alpha}=k> y^*$ for any $d_0$. This proves \cref{prop:2}(i).

Next, consider the case in which the condition of \cref{prop:2}(i) is violated. Assume the first condition holds, yet the second condition does not, that is, $\theta\geq \frac{1+\alpha}{\lambda}c_fk$ and $V'(k)\leq \eta(c_f-c_r)+(1-b) e_f \xi$. In this case, $y^*=k$ based on the analysis in the previous paragraph. Hence, setting $\bar{d}_c=0$ proves \cref{prop:2}(ii).

Assume the first condition is violated, that is, $\theta<\frac{1+\alpha}{\lambda}c_fk$, indicating that $k \left(\frac{\theta  \lambda }{(1+\alpha)  c_fk}\right)^{\frac{1+\alpha}{1+\alpha -\lambda}}< k$, i.e., $f(c_f)<1$. In this case, the welfare function $W(y)$ can be formulated as a piecewise function containing at least two pieces: %
The first piece corresponds to the $y$ range of $[0,k \left(\frac{\theta  \lambda }{(1+\alpha)  c_fk}\right)^{\frac{1+\alpha}{1+\alpha -\lambda}})$, in which $x^*=\left(\frac{\theta  \lambda }{(\alpha+1)c_f k }\right)^{\frac{1}{1+\alpha -\lambda}}$; we denote this piece by $W_1$. At least one other piece corresponds to the second or the third case of the  $x^*$ solution identified in \cref{lem:2}.
We show that $y^*$ occurs on $W_1(y)$ if and only if $d_0$ is sufficiently small. It can be calculated that  $W_1''(y)<0$, indicating that $W_1(y)$ is a concave function. Hence, if $y^*$ occurs on $W_1(y)$, then $y^*$ equals either the stationary point at which $W_1'(y)=0$ (denoted by $y_1'$) or the endpoint $y=0$. 
\begin{itemize}
\item If $V'(0)< \eta(c_f-c_r)+(1-b\cdot f(c_f)) e_f \xi$ (a condition independent of $d_0$), then $W_1(y_1')>W_1(0)$. Hence, $y^*$ occurs on $W_1(y)$ if and only if  $W(y)\leq W_1(y_1')$ holds for any $y\geq k \left(\frac{\theta  \lambda }{(1+\alpha)  c_fk}\right)^{\frac{1+\alpha}{1+\alpha -\lambda}}$. It can be calculated that for any $y\geq k \left(\frac{\theta  \lambda }{(1+\alpha)  c_fk}\right)^{\frac{1+\alpha}{1+\alpha -\lambda}}$,  $\frac{\partial(W(y)- W_1(y_1'))}{\partial d_0}$ equals either $b \xi  \left(\left(\frac{y}{k}\right)^{\frac{1}{\alpha +1}}-\left(\frac{\theta  \lambda }{(1+\alpha)c_f k}\right)^{\frac{1}{\alpha -\lambda +1}}\right)$ (if the second case of $x^*$ in \cref{lem:2} materializes, i.e., $x^*(y)=(\frac{y}{k})^{1/(1+\alpha)}$) or $b \xi  \left(\left(
\min\{\frac{\theta  \lambda }{(1+\alpha)c_rk},1\}\right)^{\frac{1}{\alpha -\lambda +1}}-\left(\frac{\theta  \lambda }{(1+\alpha)c_fk}\right)^{\frac{1}{\alpha -\lambda +1}}\right)$ (if the third case of $x^*$ in \cref{lem:2} materializes, i.e., $x^*(y)=\min\{f(c_r),1\}$). %
Both derivatives are positive for any $y\geq k \left(\frac{\theta  \lambda }{(1+\alpha)  c_fk}\right)^{\frac{1+\alpha}{1+\alpha -\lambda}}$ given $k \left(\frac{\theta  \lambda }{(1+\alpha)  c_fk}\right)^{\frac{1+\alpha}{1+\alpha -\lambda}}< k$ and  $b>0$ (and are independent of $d_0$), indicating that  $W(y)- W_1(y_1')$ is a strictly increasing linear function in $d_0$. Hence, if $y^*$ occurs on $W_1(y)$ under some $d'$, then the same happens under any $d''<d'$, indicating that either \cref{prop:2}(i) or (ii) holds. To show that \cref{prop:2}(ii) holds in this case, define $\Delta W(y)=W(y)- W_1(y_1')$ when $d_0=0$ for any $y\geq k \left(\frac{\theta  \lambda }{(1+\alpha)  c_fk}\right)^{\frac{1+\alpha}{1+\alpha -\lambda}}$. %
Define $\bar{d}'(y)=\max\{0,-\Delta W(y)/\frac{\partial(W(y)- W_1(y_1'))}{\partial d_0}\}$. By definition,  when $d_0>\bar{d}'(y)$, $W(y)>W_1(y_1')$, which implies that $y^*$ no longer occurs on $W_1(y)$. Setting $\bar{d}_c=\min\{\bar{d}'(y), \forall y\geq k \left(\frac{\theta  \lambda }{(1+\alpha)  c_fk}\right)^{\frac{1+\alpha}{1+\alpha -\lambda}}\}$  proves \cref{prop:2}(ii) in this sub-case. %

\item  If $V'(0)\geq \eta(c_f-c_r)+(1-b\cdot f(c_f)) e_f \xi$, then $W_1(y_1')\leq W_1(0)$. Replacing $y'_1$ by $y=0$ in the arguments used in the previous bullet point proves \cref{prop:2}(ii) in this sub-case. \hfill \emph{Q.E.D.}
\end{itemize}

\section{Calibration and Computational Details}
\label{ec:calibration}

This section maps observable inputs into the parameter values used in Instances A--C (Instance D combines the Instance A market size with the Instance C technology and cost parameters). We proceed in three steps. First, we estimate the market-side relationship between investment and model capability to discipline the demand parameters. Second, we calibrate the technology scaling relationship linking capability to electricity use and anchor instance-specific energy-demand levels. Third, we set energy costs, emissions intensity, and the renewable investment cost function.

\subsection{Market-Side Calibration}
\label{subsec:calib-market}

We estimate the willingness-to-pay elasticity $\lambda$ by regression. Let $w$ denote market willingness-to-pay. The power law $w=\theta x^\lambda$ is equivalent to $\log w = \lambda \log x+\log \theta$. For each instance, we assemble three annual observations of $(w,x)$ from AI investment and model capability measures for the U.S., China, and EU over 2023--2025, summarized in \cref{tab:lambda_estimate}.

\begin{table}[htbp]
\small
\centering
\caption{Investment and Model Capability by Region, 2023--2025}
\label{tab:lambda_estimate}
\begin{threeparttable}
\begin{tabular}{@{}lcccccc@{}}
\toprule
 & \multicolumn{2}{c}{United States} & \multicolumn{2}{c}{China} & \multicolumn{2}{c}{European Union} \\
\cmidrule(lr){2-3} \cmidrule(lr){4-5} \cmidrule(lr){6-7}
Year & Invest. & MMMU & Invest. & MMMU & Invest. & MMMU \\
\midrule
2025 & 285.88 & 85.4 & 76.25 & 80.1 & 20.92 & 62.5 \\
2024 & 109.08 & 78.2 &  9.29 & 70.3 & 19.42 & 52.5 \\
2023 &  67.22 & 59.4 &  7.76 & 45.9 & 11.00 & 45.6 \\
\bottomrule
\end{tabular}
\begin{tablenotes}[flushleft]
\footnotesize
\item \textit{Notes:} Investment is in billions of USD. MMMU scores reflect the top-performing model in each region-year. See text for model names and data sources.
\end{tablenotes}
\end{threeparttable}
\end{table}

We construct the investment inputs in \cref{tab:lambda_estimate} as follows. The 2023, 2024, and 2025 figures follow \citet{maslej2024aiindex}, \citet{maslej2025aiindex}, and \citet{Sajadieh2026aiindex}, respectively.\footnote{\url{https://hai.stanford.edu/ai-index/2024-ai-index-report/economy}. Figure 4.3.10.}\textsuperscript{,}\footnote{\url{https://hai.stanford.edu/ai-index/2025-ai-index-report/economy}. Figure 4.3.10. \label{fn:hai2025}}\textsuperscript{,}\footnote{\url{https://hai.stanford.edu/ai-index/2026-ai-index-report/economy}. Figure 4.2.11.\label{fn:hai2026}}
The only exception is China's investment figure: \citet{Sajadieh2026aiindex} acknowledge that their reported number may understate the investment in China. For example, the Quid database, on which the estimate in \citet{Sajadieh2026aiindex} is based, does not include internal R\&D spending of domestic tech giants. Hence, we adopt another source and estimate the investment value to be $125 \times 0.61 \approx 76.25$, where 125 billion~USD is the total AI investment in China in 2025 and 61\% is attributed to private investment.\footnote{\url{https://www.secondtalent.com/resources/chinese-ai-investment-statistics/} \label{fn:chinainvestment}} We use private investment to maintain comparability with the other investment measures, which are based on private-sector activity. We validate this estimate by collecting major Chinese companies' AI investment in 2025, including Alibaba (RMB120 billion)\footnote{\url{https://www.alibabagroup.com/en-US/document-1930116148192346112}}, ByteDance (RMB150 billion)\footnote{\url{https://www.reuters.com/technology/bytedance-plans-20-billion-capex-2025-mostly-ai-sources-say-2025-01-23}}, Tencent (RMB18 billion)\footnote{\url{https://static.www.tencent.com/uploads/2026/03/18/2804dbdae364ca25b82d21bc8304f1d3.pdf}}. Summing up these numbers and also including (partially) the R\&D expenditure of other companies that have heavily invested in AI (such as Huawei, Xiaomi, JD.com, and Baidu)\footnote{The total R\&D expenditures of these companies are Huawei (\href{https://www.huawei.com/en/annual-report}{RMB192.3 billion}), Xiaomi (\href{https://ir.mi.com/system/files-encrypted/nasdaq_kms/assets/2026/04/28/5-29-08/Xiaomi\%202025\%20AR_EN.pdf}{RMB33.1 billion}), JD.com (\href{https://ir.jd.com/news-releases/news-release-details/jdcom-announces-fourth-quarter-and-full-year-2025-results-and}{RMB22.2 billion}), and Baidu (\href{https://ir.baidu.com/news-releases/news-release-details/baidu-announces-fourth-quarter-and-fiscal-year-2025-results}{RMB20.4 billion}). Assuming 35\%--65\% of these R\&D investments go to AI leads to the estimated range reported above.}, we estimate the corporate investment in AI in China in 2025 to be approximately RMB380--460 billion (\$56-\$68 billion). Combining this number with the Quid-based figure reported by \citet{Sajadieh2026aiindex} gives \$68.41--80.41 billion, which is consistent with the \$76.25 billion figure estimated from the other source. \looseness=-1

Model capability is measured by MMMU. The MMMU scores adopted for Instances A and B in \cref{tab:lambda_estimate} are based on the highest MMMU (Val) evaluation reported for U.S.- and China-developed models in the corresponding year.\footnote{\url{https://mmmu-benchmark.github.io/}}
The MMMU scores in Instance C are based on the performance of Mistral AI models reported in 2024\footnote{\url{https://mistral.ai/news/pixtral-12b}} and 2025.\footnote{\url{https://llm-stats.com/models/mistral-small-3.2-24b-instruct-2506}} 
No MMMU scores have been reported for Mistral's 2023 models such as Mistral 7B Instruct. Mistral 7B Instruct achieved a score of 60.1 in the MMLU benchmark,\footnote{\url{https://docsbot.ai/models/compare/mistral-large/mistral-7b-instruct}} while Pixtral 12B achieved 69.2.\footnote{\url{https://llm-stats.com/models/pixtral-12b-2409}} Assuming the same percentage increase in MMLU and MMMU scores between these two models, we estimate the MMMU score of Mistral 7B Instruct to be 45.6. 

We set $w$ to be the investment amounts and $x$ to be the MMMU scores (normalized by 100 to lie in its range $[0,1]$), and perform a linear regression for each instance. This yields a $\lambda$ value of 3.46, 3.19, and 1.98. %

Three observations per region leave one residual degree of freedom, so these point estimates carry substantial statistical uncertainty. \cref{tab:lambda_ci} reports the OLS standard errors and 95\% confidence intervals: no regional estimate rejects either scaling regime at conventional levels (each $|t|<0.7$ against the estimated threshold $1+\alpha\approx 2.47$), and the threshold estimate itself is imprecise. Two measurement caveats compound the statistical one: the 2025 China figure is taken from an alternative total-investment source and validated against an AI-Index-based component combined with a corporate-investment estimate constructed from company disclosures (see above), and compression of MMMU scores near the benchmark's ceiling mechanically understates late capability growth, inflating both estimated exponents. We therefore treat the estimates as anchors for scenario construction rather than as measured classifications.

\begin{table}[htbp]
\small
\centering
\caption{Uncertainty and Alternative Estimators for the Market-Scaling Exponent}
\label{tab:lambda_ci}
\begin{threeparttable}
\begin{tabular}{@{}lccccc@{}}
\toprule
 & $\hat\lambda$ & (s.e.) & 95\% CI & Investment-based & Price-based\tnote{a} \\
\midrule
United States  & 3.46 & (1.79) & $(-19.3,\ 26.2)$ & market-led   & market-led \\
China          & 3.19 & (2.98) & $(-34.6,\ 41.0)$ & market-led   & resource-led \\
European Union & 1.98 & (1.02) & $(-11.1,\ 15.0)$  & resource-led & market-led \\
\midrule
Threshold $1+\alpha$ & 2.47 & (0.42) & $(-2.9,\ 7.8)$ & \multicolumn{2}{c}{--} \\
\bottomrule
\end{tabular}
\begin{tablenotes}[flushleft]
\footnotesize
\item \textit{Notes:} OLS on three annual observations per regression (one residual degree of freedom); confidence intervals use $t_{0.975,1}=12.71$.
\item[a] Classifications implied by fitting model prices over time, as examined in a conference discussion of this paper and reported here as indicative pending a formal estimation; within-year price-on-capability regressions are sample-dependent (below the threshold for the most-used models, near it for frontier models).
\end{tablenotes}
\end{threeparttable}
\end{table}

A further identification concern is that investment proxies the level of anticipated willingness to pay rather than the willingness-to-pay function itself. Writing the proxy as $w_t=\theta_t x_t^{\lambda}$, the time-series slope estimates $\lambda+\Delta\log\theta_t/\Delta\log x_t$, and the 2023--2025 expansion of market breadth ($\theta_t$) biases investment-based estimates upward. Price-based estimators drift in the opposite direction under intensifying competition and penetration pricing, and within-year price-on-capability regressions are unstable because current prices are strategic objects: frontier models are often free to consumers, and less expensive models sometimes outperform more expensive ones. As \cref{tab:lambda_ci} reports, willingness-to-pay estimators built from model prices reverse the China and EU classifications while agreeing with the U.S.\ classification. We conclude that the U.S.\ market-led classification is robust across estimators, while the China and EU classifications are not determined by presently available data. This does not disturb the analysis: the four case-study instances span both scaling regimes by construction, so reclassifying a region relocates it among instances whose equilibria are already characterized. 

Looking ahead, $\lambda$ becomes identifiable as the market matures: price-capability gradients along iso-share curves (hedonic regressions that control for market position), or discrete-choice estimation on model usage shares, prices, and capability indices, recover willingness to pay for capability from substitution patterns rather than from spending. We view this as a measurement agenda that the model motivates once prices decouple from investor subsidies.

The regression analysis also provides an estimate of $\theta$ (\$371.24, \$74.44, \$57.57 billion for the three instances, respectively), but it aggregates market conditions across 2023--2025 and does not target the 2024 conditions that underlie the remaining parameter choices. We therefore anchor $\theta$ to the 2024 investment levels. A literal application of $w=\theta x^\lambda$ would recover $\theta$ as $\frac{w}{x^\lambda}$. We do not take this route because aggregate AI investment is best viewed as a proxy for the scale of anticipated market opportunity rather than as a direct observation of $\theta x^\lambda$. This difference affects the slope ($\lambda$) and level ($\theta$) parameters differently. For example, multiplying all investment observations by the same factor does not change the $\lambda$ estimate, whereas it can scale the $\theta$ estimate proportionally. Hence, directly dividing investment figures by $x^\lambda$ would impose more structure than the data can support. Instead, anchoring $\theta$, a market-potential parameter,  directly to observed 2024 investment yields a transparent and economically interpretable measure of market scale. Accordingly, we set the $\theta$ values in Instances A and C to be the amounts of private AI investment in the U.S.\ and the EU in 2024 reported by \citet{maslej2025aiindex}.\footref{fn:hai2025} 
For Instance B, we augment the estimate by \citet{maslej2025aiindex} with a potential corporate-investment component. Specifically, we adopt the lower-end estimate of the 2025 corporate investment (\$56 billion, as estimated above) and scale it down using the 129.9\% growth in total AI investment between 2024 and 2025 reported by \citet{Sajadieh2026aiindex},\footref{fn:hai2026} which implies a corporate AI investment potential of approximately \$24.36 billion in 2024. Adding this amount to the Quid-based estimate of \$9.29 billion yields a total investment potential estimate of approximately \$33.65 billion.%

There are multiple estimates of the total economic value of Generative AI. %
It has been widely agreed that AI can increase productivity and contribute to economic growth in other industries. Goldman Sachs estimated that AI could lead to \$7 trillion in global GDP growth over 10 years.\footnote{\url{https://www.goldmansachs.com/insights/articles/generative-ai-could-raise-global-gdp-by-7-percent}} 
\citet{acemoglu2025simple} provided a lower estimate of a 1.6\% increase in global GDP for the same 10-year horizon, or about \$183.78 billion per year. We take the average of these two estimates and calculate that the annual economic spillover of AI is approximately \$441.89 billion. We divide that number by the revenue potential of the AI industry. To match the economic value potential estimate that aggregates multi-year conditions, we adopt the regression-based $\theta$ estimate reported in the previous paragraph as a reasonable proxy of the annual revenue potential of the AI industry within the comparable time frame. This yields $\eta = 441.89/(371.24+74.44+57.57)+1\approx 1.88$.

The next subsection calibrates the technology scaling relationship between capability and energy use.

\subsection{Scaling and Energy Demand}
\label{subsec:calib-scaling}

We calibrate the technology scaling parameter $\alpha$ via regression. Let $v$ denote the AI energy usage. The function $v=kx^{1+\alpha}$ is equivalent to the linear function $\log v = (1+\alpha)\log x +\log k$. We consider three sets of $(v,x)$ estimates in 2024, summarized in \cref{tab:alpha_estimate}. 

\begin{table}[htbp]
\centering
\small
\caption{Data-Center Electricity Consumption and Model Capability by Region, 2024}
\label{tab:alpha_estimate}
\begin{tabular}{@{}lccc@{}}
\toprule
 & United States & China & European Union \\
\midrule
Data-center electricity (TWh) & 177.51 & 115.42 & 63.51 \\
MMMU score               &  78.2  &  70.3  & 52.5  \\
\bottomrule
\end{tabular}
\end{table}

The energy consumption figures are based on estimates reported in Figure 2.12 by the International Energy Agency (IEA 2025)\footnote{\url{https://iea.blob.core.windows.net/assets/de9dea13-b07d-42c5-a398-d1b3ae17d866/EnergyandAI.pdf} \label{fn:iea}}. These figures measure data-center electricity consumption across all workloads rather than AI alone: AI is the principal driver of recent growth, but it is one workload among several. We therefore treat them as anchoring magnitudes for the energy-demand scale $k$ rather than as measurements of AI-only demand. %
The MMMU scores are the same as those shown in \cref{tab:lambda_estimate}. Linear regression based on these estimates yields an $\alpha \approx 1.467$. This three-point estimate carries the same caveat as the $\lambda$ regressions (\cref{tab:lambda_ci}). The power-law form and curvature of the technology side are, however, documented independently and densely in the scaling-law literature \citep{kaplan2020scaling,hoffmann2022train}, whereas the market exponent has no comparable corroboration; we therefore concentrate scenario uncertainty in $\lambda$. Two further caveats apply. First, the procedure is two-step: the pooled regression identifies the slope under a common intercept, and the region-specific $k$ values are then re-anchored, so $\alpha$ is best read as a common-curvature assumption disciplined by cross-regional variation rather than a jointly identified parameter. Second, log-log slopes depend on the benchmark's origin: MMMU is multiple-choice with a chance floor near 25, and re-anchoring scores at that floor yields materially smaller exponents (about 1.5 rather than 2.5 for $1+\alpha$, with $\lambda$ falling commensurately) while preserving all three regime labels.

Although the regression analysis also provides a $k$ estimate, it does not account for the regional energy usage differences that we intend to capture through the different instances in the case study. To reflect such heterogeneity, we anchor the $k$ estimate in each instance based on the actual energy consumption in the corresponding region. Specifically, we calculate $k$ as $\frac{v}{x^{1+\alpha}}$ where $v$ and $x$ are the energy consumption and MMMU score shown in \cref{tab:alpha_estimate}, and obtain 325.6 TWh, 275.32 TWh, and 311.32 TWh for Instances A--C, respectively.  %

The next subsection sets energy cost, emissions, and renewable investment parameters.

\subsection{Energy Costs, Emissions, and Renewable Investment}
\label{subsec:calib-energy}

We set $c_r$ to be the lowest cost of a renewable PPA with 80\% hourly matching, which is approximately 50, 48, and 65 USD per MWh according to Figure 2.19 by IEA (2025).\footref{fn:iea}

We estimate $c_f$ based on the industrial retail price range of grid electricity reported in Figure 2.19 by IEA (2025).\footref{fn:iea} This range is about $[79.27,93.90]$, $[87.80,150]$, and $[109.76,200]$ for the U.S., China, and EU, respectively (all numbers are in USD/MWh). Considering that the European range covers the entire EU and thus may be an overestimate of the volatility of the industrial rates in data-center hubs, we instead focus on the grid rates in Frankfurt, which has the highest concentration of data centers in Europe. The wholesale rates to data centers in Frankfurt  are reported to be €110--€135/MWh,\footnote{\url{https://techplustrends.com/ai-data-center-energy-cost-europe-market-guide/}. We also find that the result is insensitive to grid rate ranges as long as they are above a certain threshold. Hence, results based on the rates in Frankfurt generally apply to other hubs (FLAP-D).} which lead to a range in USD/MWh of approximately $[125.13,153.63]$. %
Consistent with the model's interpretation of $c_f$ and $c_r$ as unit generation costs, we proxy both by observed procurement prices (grid rates and PPA prices, respectively). Because the two sources are proxied symmetrically and the analysis operates mainly through the difference $c_f-c_r$, a common retail margin largely nets out. We assume a uniform distribution of the grid rate, denoted by $r_g$, in the ranges defined above. To reflect developers' risk aversion to grid rate volatility, we adopt a mean-variance cost based on the following formula: $c_f=\mathbb{E}[r_g]+0.1\,\operatorname{Var}(r_g)$, where the coefficient $0.1$ carries units of MWh per USD so that the adjustment is expressed in USD/MWh. Accordingly, we can calculate $c_f$ to be approximately 88, 151, and 146 USD per MWh in the three instances, respectively. %

We next calibrate the renewable investment cost function $V(y)$. In practice, renewable-investment costs may exhibit different curvature across the deployment path: learning, standardization, and economies of scale can initially reduce marginal costs, whereas increasingly constrained siting, longer interconnection distances, and the need for transmission and grid expansion can eventually cause marginal costs to rise. Because our analysis focuses on the additional renewable capacity required to support future AI growth, we model only the emerging convex segment. We adopt a power-law formula for this cost function: $V(y)=g\cdot y^\mu$. The calibration is based on solar farm investment costs in practice. To estimate $\mu$, we collect investment information on major solar projects in China in recent years, summarized in \cref{tab:solar}. Accordingly, it can be calculated that the average investment cost of existing solar capacities (i.e., those of the first three projects) is 0.60 billion USD/GW. The Ordos project will increase the total capacity by 94\% and, at the same time, raise the average cost to 0.76 billion USD/GW, representing a 25\% increase. Given cost function $V(y)=g\cdot y^\mu$, the average investment cost is $g\cdot y^{\mu-1}$. Hence, the data collected implies $1.94^{\mu-1}=1.25$. Solving this equation yields $\mu\approx 1.34$. %
\begin{table}[htbp]
\small
\centering
\caption{Major Solar Projects in China}
\label{tab:solar}
\begin{threeparttable}
\begin{tabular}{@{}lcccc@{}}
\toprule
 & Midong & Mengxi Lanhai & Tianwan & Ordos \\
\midrule
Capacity (GW)            & 3.5  & 3.0  & 2.0  & 8.0            \\
Investment (billion USD) & 2.13 & 1.60 & 1.40 & 7.33\tnote{*}  \\
Grid connection          & 2024 & 2024 & 2025 & 2027 (expected) \\
\bottomrule
\end{tabular}
\begin{tablenotes}[flushleft]
\footnotesize
\item[*] Investment estimated based on total project cost of USD 11 billion, allocated by capacity share (8 GW solar, 4 GW wind).
\item \textit{Sources:} PV Magazine (Midong,\footnote{\url{https://www.pv-magazine.com/2024/06/06/worlds-largest-solar-plant-goes-online-in-china-2/}} Mengxi Lanhai\footnote{\url{https://www.pv-magazine.com/2024/11/12/worlds-second-largest-solar-plant-goes-online-in-china/}}); Saur Energy\footnote{\url{https://www.saurenergy.com/solar-energy-news/construction-of-2gw-offshore-solar-plant-in-china-starts}} (Tianwan); Reuters\footnote{\url{https://www.reuters.com/business/energy/china-three-gorges-renewables-plans-11-bln-new-energy-project-inner-mongolia-2024-06-28/}} (Ordos).
\end{tablenotes}
\end{threeparttable}
\end{table}

Renewable investment, from a policymaker's perspective, includes not only the installation cost of the plant and on-site infrastructure, but also the broader system investments required to connect and deliver its output, including transmission expansion, grid reinforcement, land, rights-of-way, and financing. Accordingly, we estimate $g$ based on a  calibration of the all-in public-system cost of utility-scale solar instead of the directly observed installation cost. We begin with the U.S.\ Energy Information Administration's capacity-weighted average construction cost of \$1.865 million/MW for solar generators entering service in 2024.\footnote{\url{https://www.eia.gov/electricity/generatorcosts/}. This measure already includes site preparation, equipment, engineering and construction management, permitting and other owner costs, and the direct electrical interconnection to a nearby transmission system, but excludes land acquisition and financing; hence, we do not add a separate project-level interconnection charge.}  To capture broader transmission investment beyond the immediate grid tie-in, we assume that each MW requires an attributable 100 MW-mile of additional transfer capability and apply the midpoint (approximately \$5,500/MW-mile) of the Department of Energy's geographically varying assumptions of roughly \$2,000--\$9,000/MW-mile,\footnote{\url{https://www.energy.gov/sites/default/files/2024-10/NationalTransmissionPlanningStudy-Chapter2.pdf}. Figure A-1.} yielding an allowance of \$0.55 million/MW; this assumption is intended to reflect that marginal projects are increasingly located farther from load centers and existing transmission. We further assume eight acres/MW, within DOE's six-to-eight-acre planning range,\footnote{\url{https://www.energy.gov/sites/default/files/2021-12/developing-utility-scale-renewable-electricity.pdf}} and value the land at the 2025 national average farm-real-estate value of \$4,350 per acre,\footnote{\url{https://www.nass.usda.gov/Publications/Todays_Reports/reports/land0825.pdf}} adding approximately \$0.035 million/MW. Finally, because the EIA construction-cost measure excludes financing, we add a 10\% construction-finance allowance to the plant cost and a 5\% system-delivery contingency to the plant, transmission, and land subtotal to cover uncertainty in regional grid reinforcement, transmission rights-of-way, and public planning and mitigation requirements. The resulting estimate is $1.865+0.550+0.035+1.865\times 10\%+(1.865+0.550+0.035)\times 5\%\approx$ \$2.8 million/MW, or equivalently \$2.8 billion/GW.\footnote{The estimate excludes storage, ongoing operation and maintenance, and decommissioning costs.}

To translate this investment cost into the coefficient $g$ in $V(y)=g y^{1.34}$, we normalize the marginal investment cost at $y=1$ GW to the all-in unit cost of \$2.8 billion/GW estimated above. Because $V'(y)=1.34g y^{0.34}$, imposing $V'(1)=2.8$ yields $g=2.8/1.34\approx 2.089$ billion USD per $\mathrm{GW}^{1.34}$. We use 1 GW as the normalization point, rather than applying the calibration to the entire existing solar capacity, because the exponent 1.34, calibrated from projected changes in solar investment requirements between 2025 and 2027, is intended to characterize the local curvature of the emerging convex-cost segment rather than the full historical deployment path. A relatively low normalization point therefore better represents the onset of rising marginal costs associated with increasingly difficult siting, interconnection, and grid expansion to meet the AI energy demand. We also cross-validate the resulting annualized cost against prevailing estimates of the levelized cost of solar electricity and find that the calibration is of a comparable magnitude.\footnote{Specifically, this function we derive, $V(y)=0.049\cdot y^{1.34}$, implies a cost of 0.049 billion at 1TWh, which translates into \$49 per MWh, which is comparable with the prevailing range of the levelized cost of utility-scale solar in the U.S.} Assuming a 20-year economic life, the amortized annual capacity-based cost function becomes $0.104\cdot y^{1.34}$ billion USD per year, with $y$ measured in GW. The amortization is undiscounted; incorporating a cost of capital would raise $g$ and hence $V'(k)$, tightening condition (b) in \cref{prop:1}(i) and making the adaptation trap more likely rather than less. At a capacity factor of 20\%, one GW of installed solar capacity generates $8.76\times0.20=1.752$ TWh annually. Substituting $y_{\mathrm{GW}}=y_{\mathrm{TWh}}/1.752$ therefore gives an equivalent annual energy-based cost function of $0.049\cdot y^{1.34}$, where $y$ is measured in TWh. This is the $V(y)$ function used in Instance A.

To calibrate the coefficient $g$ for Instances B and C, we assume that the all-in solar investment cost across regions scales proportionally with the prevailing cost of utility-scale solar installation. Based on available industry estimates, we use installation costs of approximately 0.95, 0.57, and 0.70 USD/W for the U.S., China, and EU, respectively.\footnote{\url{https://a1solarstore.com/blog/what-is-a-solar-farm-and-how-much-money-can-it-make-you.html} \url{https://compareelectricity.com/research/how-much-does-a-solar-farm-cost}}$^,$\footnote{\url{https://www.solarpowereurope.org/press-releases/new-study-reveals-path-to-reshore-solar-manufacturing-in-europe}. This article estimates that the cost of utility solar installations in Europe is about 60.8 €ct/Wp compared to 50.0 €ct/Wp for a Chinese system.}
We therefore scale the benchmark coefficient $g=0.049$ for Instance A by the corresponding installation-cost ratios. This yields $g\approx 0.029$ and $g\approx 0.036$ in Instances B and C, respectively.

\enlargethispage{3\baselineskip}
Finally, we set $e_f$ in Instances A--C to be the carbon intensity of the grid in the U.S., China, and EU, respectively, which is approximately 0.367 ton/MWh (calculated by dividing the total emissions of approximately 1.53 billion metric tons by the total generation of approximately 4.18 trillion kilowatthours\footnote{\url{https://www.eia.gov/tools/faqs/faq.php?id=74&t=11}}), 0.614 ton/MWh,\footnote{\url{https://academic.oup.com/ijlct/article/doi/10.1093/ijlct/ctae181/7762370}} and 0.187 ton/MWh.\footnote{See Figure~1 at \url{https://www.eea.europa.eu/en/analysis/indicators/greenhouse-gas-emission-intensity-of-1}.}
We set the climate-damage weight $\xi$ to the carbon price in the corresponding regions. In the U.S., the carbon price tracks around \$28 per ton for regional cap-and-trade systems such as California's program;\footnote{\url{https://icapcarbonaction.com/en/ets/usa-california-cap-and-invest-program}} legislative frameworks such as the Whitehouse-Schatz bill targeted a starting fee of \$52 per ton.\footnote{\url{https://www.c2es.org/wp-content/uploads/2021/12/carbon-pricing-proposals-in-the-117th-congress.pdf}} Taking the average of these estimates, we set $\xi=\$40$ per ton in Instance A. The carbon price in China's national Emissions Trading Scheme hovers around 90 to 106 RMB (approximately 12 to 15 USD) per ton.\footnote{\url{https://www.ieta.org/uploads/wp-content/Resources/Busines-briefs/2025/IETA_Business_Brief-China-final-July2025.pdf}} Taking the average of these estimates, we set $\xi=\$13.5$ per ton in Instance B. The carbon price in the EU is approximately €83.5 (\$95) per ton,\footnote{\url{https://tradingeconomics.com/commodity/carbon} (accessed July 2026)} which we use as the $\xi$ value in Instance C. We anchor $\xi$ to observed carbon prices rather than to social-cost-of-carbon estimates. Because permit prices are policy-constrained and typically run well below the social cost of carbon, this is a conservative valuation of climate damages. %

\end{document}